\documentclass[12pt]{article}
\usepackage[utf8]{inputenc}
\usepackage[T1]{fontenc}
\usepackage[english]{babel}
\usepackage{amsmath, amssymb}
\usepackage{amsthm}
\newtheorem{proposition}{Proposition}
\usepackage{geometry}
\usepackage{array,booktabs,multirow}
\usepackage{tikz}
\usetikzlibrary{calc, positioning}
\usepackage[american]{circuitikz} 
\usepackage{newtxtext,newtxmath}  
\usepackage{subcaption}
\usepackage{enumitem}
\usepackage{setspace}
\usepackage{parskip}
\usepackage[table]{xcolor}
\usepackage{wasysym}
\usepackage{multicol}
\usepackage{longtable}
\usepackage{amsfonts}
\usepackage{amsfonts}
\usepackage{fancybox}
\usepackage{caption}
\usepackage{qtree}
\usepackage{tikz}
\usepackage{circuitikz}
\usetikzlibrary{calc}
\usetikzlibrary{shapes,snakes}
\usepackage{relsize} 
\usepackage{csquotes}
\usepackage[hidelinks]{hyperref}
\usepackage[
  backend=biber,
  style=ieee,   
  sorting=nyt,        
  maxbibnames=99
]{biblatex}
\tikzset{
  dot/.style={circle, fill=black, inner sep=1.2pt},
  edgelabel/.style={midway, sloped, above=2pt, font=\Large},
  smalllabel/.style={font=\scriptsize, inner sep=1pt},
}

\newcommand{\TriCoords}{%
  \coordinate (i) at (0,0);
  \coordinate (j) at (3,0);
  \coordinate (k) at (1.5,2.6);
}

\newcommand{\DiagA}{%
\begin{tikzpicture}
  \TriCoords
  \draw (i) -- node[midway, font=\footnotesize, sloped=false, xshift=-6pt] {P} (k)
        (k) -- node[midway, font=\footnotesize, sloped=false, xshift= 8pt] {S} (j)
        (i) -- node[edgelabel, font=\footnotesize, below=2pt, sloped] {T} (j);
  \node[dot, label={[font=\footnotesize]below left:{i}}]  at (i) {};
  \node[dot, label={[font=\footnotesize]below right:{j}}] at (j) {};
  \node[dot, label={[font=\footnotesize]above:{k}}]       at (k) {};
\end{tikzpicture}
}

\newcommand{\DiagB}{%
\begin{tikzpicture}
  \TriCoords
  \draw (i) -- node[midway, font=\footnotesize, sloped=false, xshift=-6pt] {P} (k)
        (k) -- (j)
        (i) -- node[edgelabel, font=\footnotesize, below=2pt, sloped] {T} (j);
  \node[dot, label={[font=\footnotesize]below left:{i}}]  at (i) {};
  \node[dot, label={[font=\footnotesize]below right:{j}}] at (j) {};
  \node[dot, label={[font=\footnotesize]above:{k}}]       at (k) {};
  \path (k) -- coordinate[midway] (midpoint) (j);
  \node[dot] at (midpoint) {};
  \path (k) -- node[pos=0.20, xshift=0.40cm, font=\footnotesize] {sp} (j);
  \path (k) -- node[pos=0.72, xshift=0.40cm, font=\footnotesize] {ts} (j);
\end{tikzpicture}
}

\newcommand{\DiagC}{%
\begin{tikzpicture}
  \TriCoords
  \path (k) -- coordinate[midway] (midpoint) (j);
  \node[dot] at (midpoint) {};
  \draw (k) -- (j);
  \draw (i) to[bend right=30] node[midway, below=6pt, font=\footnotesize] {T} (midpoint);
  \draw (i) to[bend left=30]  node[midway,  above=5pt, font=\footnotesize] {P} (midpoint);
  \draw (k) -- (midpoint) node[midway,  xshift=0.45cm, font=\footnotesize] {sp};
  \draw (midpoint) -- (j) node[midway, xshift=0.32cm, yshift=3pt, font=\footnotesize] {ts};
  \node[dot, label={[font=\footnotesize]below left:{i}}]  at (i) {};
  \node[dot, label={[font=\footnotesize]below right:{j}}] at (j) {};
  \node[dot, label={[font=\footnotesize]above:{k}}]       at (k) {};
\end{tikzpicture}
}

\newcommand{\DiagD}{%
\begin{tikzpicture}
  \coordinate (c) at (1.5,1.2);
  \coordinate (i) at (0,0);
  \coordinate (j) at (3,0);
  \coordinate (k) at (1.5,2.6);

  \draw (c) -- node[midway, font=\footnotesize, xshift=-7pt, yshift=8pt] {TP} (i)
        (c) -- (j) node[midway, font=\footnotesize, xshift=5pt,  yshift=10pt] {ts}
        (c) -- (k) node[midway, font=\footnotesize, xshift=12pt] {sp};

  \node[dot] at (c) {};

  \node[dot, label={[font=\footnotesize]below left:{i}}]  at (i) {};
  \node[dot, label={[font=\footnotesize]below right:{j}}] at (j) {};
  \node[dot, label={[font=\footnotesize]above:{k}}]       at (k) {};
\end{tikzpicture}
}

\newcommand{\Ltr}[1]{%
  \ifnum#1=1 a\else
  \ifnum#1=2 b\else
  \ifnum#1=3 c\else
  \ifnum#1=4 d\else
  e\fi\fi\fi\fi
}

\newcommand{\Up}[1]{%
  \ifnum#1=1 A\else
  \ifnum#1=2 B\else
  \ifnum#1=3 C\else
  \ifnum#1=4 D\else
  E\fi\fi\fi\fi
}
\newcommand{\Lo}[1]{%
  \ifnum#1=1 a\else
  \ifnum#1=2 b\else
  \ifnum#1=3 c\else
  \ifnum#1=4 d\else
  e\fi\fi\fi\fi
}

\newcolumntype{Y}{>{\centering\arraybackslash}m{3.0cm}} 
\newcolumntype{C}{>{\centering\arraybackslash}m{1.2cm}} 
\newcolumntype{E}{>{\centering\arraybackslash}m{3.2cm}} 
\newcolumntype{S}{>{\centering\arraybackslash}m{1.5cm}} 

\newcommand{\pl}{\mathbin{\rotatebox[origin=c]{90}{$\eqcirc$}}}

\newcommand{\p}{\mathbin{p}}
\newcommand{\s}{\mathbin{s}}
\newtheorem{example}{Example}

\title{A Perspective on the Algebra, Topology, and Logic of Electrical Networks}
\author{
Marko Orešković$^{1}$, Ivana Kuzmanović Ivičić$^{2}$\thanks{Corresponding author}, Juraj Benić$^{2}$, Mario Essert$^{3}$\\[1ex]
\small $^{1}$National and University Library in Zagreb, Croatia\\
\small $^{2}$School of Applied Mathematics and Informatics, University of Osijek, Croatia\\
\small $^{3}$Retired, formerly Faculty of Electrical Engineering and Computing, University of Zagreb\\[1ex]
\small \texttt{moreskovic@nsk.hr, ikuzmano@mathos.hr, jbenic@mathos.hr, messert@inet.hr }
}
\date{\today}

\begin{document}

\maketitle

\begin{abstract}
This paper presents a unified algebraic, topological, and logical framework for electrical one-port networks based on Šare’s \(m\)-theory. Within this formalism, networks are represented by \(m\)-words (jorbs) over an ordered alphabet, where series and parallel composition induce an \(m\)-topology on \(m\)-graphs with a theta mapping \(\vartheta\) that preserves one-port equivalence. The study formalizes quasi-orders, shells, and cores, showing their structural correspondence to network boundary conditions and impedance behavior. The \(\lambda\text{--}\Delta\) metric, together with the valuation morphism \(\Phi\), provides a concise descriptor of the impedance-degree structure.  
In the computational domain, the framework is extended with algorithmic procedures for generating and classifying non-isomorphic series–parallel topologies, accompanied by programmatic Cauer/Foster synthesis workflows and validation against canonical examples from Ladenheim’s catalogue. The resulting approach enables symbolic-to-topological translation of impedance functions, offering a constructive bridge between algebraic representation and electrical realization. Overall, the paper outlines a self-consistent theoretical and computational foundation for automated network synthesis, classification, and formal verification within the emerging field of Jorbology.
\end{abstract}
\section{Introduction}

Miroslav {\v S}are (1918--2005) discovered in the spring of 1957, within the
theory of linear electrical networks, the possibility of characterising the
structure of CRL two--terminal networks by words formed from a three--element alphabet,
perceiving algebraic
operations in the series and parallel connection of two--terminals. After thirteen years of work he presented, in the monograph
\emph{$m$--Brojevi} \cite{Sare1970}, an axiomatic characterisation of the discovered algebraic
structure, called the \emph{$m$--structure}, together with its application to the
theory of RLC networks.

The $m$--structure, initially referring only to so--called two--generator networks (i.e.\ those composed of just two kinds of elements such as inductances and capacitances, or resistances and capacitances, or resistances and inductances), was further elaborated in his doctoral dissertation \cite{Sare1972}, submitted at the Faculty of Electrical Engineering in Zagreb in 1973, where it was shown that the theory inaugurated new disciplines within network theory.

From that time until his advanced age the author succeeded in generalising the
two--generator $m$--structure to a structure with an arbitrary number of
generators. This work was published in December~2000 in the book entitled
\emph{Jorbologija}\footnote{``Rebmunology'': the “jorb” being the Croatian word \emph{broj} (“number”) read backwards.}\cite{Sare2000}. Only with the publication in a scientific journal \cite{Essert2017} did it become clear that the \(m\)-theory has a broader scope, just as the author had anticipated in his book.

The purpose of this article is twofold: first, to outline a mathematical formalism and to present an illustrative algorithm that serves as a starting point for generating topologically distinct realizations of one-port networks, providing a foundation for the development of more advanced methods in the future. This approach aims to contribute towards the improvement and extension of existing catalogs, such as Ladenheim’s, by moving in the direction of a more systematic classification of topological equivalents corresponding to a given impedance function, and by suggesting new synthesis procedures that may gradually bridge the gap between theoretical considerations and practical network design. Second, the article will provide a concise yet essential overview of $m$-theory, with the intention of fostering its further development.

The paper is organized as follows. Chapter~2 provides a general algebraic introduction and establishes the fundamental notation and conventions adopted throughout the text. 

Chapter~3 introduces the concept of \emph{jorbs}\footnote{The same mathematical construct is referred to as an \emph{\(m\)-word} when the emphasis is on symbolic concatenation and string structure; as an \emph{\(m\)-number} when its algebraic properties—such as operators, quasi-order, and structural relations—are under consideration; and as a \emph{jorb} when applied specifically to electrical networks.} (\emph{m-words}, \emph{m-numbers}) as abstract algebraic representations of electrical networks. 

Chapter~4 reviews the logical foundations relevant to the proposed framework, albeit within a conceptual setting distinct from that adopted in~\cite{Essert2017}. The subsequent chapters focus on the interpretation of the jorb as a symbolic representative of the network impedance~$Z(s)$. 

Chapter~6 addresses the application of computational methods for the generation of non-isomorphic impedance structures using jorbs and discusses potential improvements in the classification of electrical networks with respect to Ladenheim’s catalogue. 

Chapter~7 presents methodologically distinct approaches to the synthesis of one-port networks formulated through the jorb formalism. 

Finally, Chapter~8 and Chapter~9 outlines prospective directions for further research, emphasizing unresolved problems in the emerging field of \emph{Jorbology} and its potential applications in the broader context of network theory.

\section{Algebra of $m$‑Numbers}
	The Šare's \textbf{$m$-system}
	is a quadruple  ($\Gamma$, $<_\Gamma$, $\cdot$, $M_\Gamma$), where:
	\begin{itemize}
		\item set $\Gamma$ is an \textit{finite alphabet} (also caled $m$-alphabet) which consists of \textit{symbols}, that are \textit{totally ordered} by relation  $<_\Gamma$,
		\item '$\cdot$' denotes the \textit{concatenation} operation on symbols from $\Gamma$,
		\item  set $M_\Gamma$ consists of all words obtained by concatenation of an \textit{even number of symbols} from $\Gamma$ (also called \textit{$m$-words}). The concatenation of two \textit{$m$-words} also yields an \textit{$m$-word}.
	\end{itemize}

\subsection{The alphabet}

For an alphabet of $n$ symbols, which are usually represented by lowercase letters of the English alphabet, we can write: \(\Gamma_n = \{a_1, a_2, \ldots, a_n\}\) where \(a_1 <_\Gamma a_2 <_\Gamma \cdots <_\Gamma a_n\). The initial symbol of the alphabet will be denoted by \(\alpha(\Gamma_n)\), and the last one by \(\omega(\Gamma_n)\). 
The order of symbols, i.e., the relation relation $<_\Gamma$ is induced by the \textit{valuation function}:
\[
v: \Gamma \rightarrow \mathbb{Z}, \qquad a_i <_{\Gamma} a_j \;\;\Longleftrightarrow\;\; v(a_i) < v(a_j). 
\]
If $\Gamma \subset \mathbb{Z}$, then $v$ is usually the identity function. This function allows for determining the larger or lower symbol when comparing two symbols:

\begin{itemize}
	\item \textit{min/lower}  - for any $a,b\in\Gamma$: $$a\downarrow b:=\left\{\begin{array}{ll}
	a,& a<_{\Gamma} b,\\
	b,&\mbox{otherwise}.
	\end{array}\right.$$
	
	\item \textit{max/larger}  - for any $a,b\in\Gamma$: $$a\uparrow b:=\left\{\begin{array}{ll}
	b,& a<_{\Gamma} b,\\
	a,&\mbox{otherwise}.
	\end{array}\right.$$
\end{itemize}

\begin{example}\label{ex:jorb}
\leavevmode
\begin{itemize}
	\item $\Gamma_2=\{\beta, \psi\}$,  where $v(\beta)<v(\psi)$,  i.e.  $\beta<_{\Gamma}\psi$
	\\ $\psi\downarrow\beta=\beta$,\quad  $\psi\uparrow\beta=\psi$
	\item $\Gamma_3=\{a,b,c\}$, where the valuation is given by
$v(a) = -1$, $v(b) = 0$, $v(c) = 1$. 
Then, with respect to the induced order $<_\Gamma$, we have:\\
$a \uparrow c = c$,\quad $b \downarrow c = b$
	\item $\Gamma_5=\{1,2,3,4,5\}$, where it is $1<_\Gamma 2<_\Gamma 3<_\Gamma 4<_\Gamma 5$, we have:\\ $2 \uparrow 4 =4$ and $2 \downarrow 4 =2$.
\end{itemize}
\end{example}

The duality is a key concept in the $m$-alphabet, where each symbol has a corresponding dual with respect to the center or the axis of symmetry of the alphabet. If the alphabet has an odd number of symbols, then the symbol at the position of the middle element $\lceil \frac{n}{2} \rceil$ is self-dual. 

Symbols \(a_k\) and \(a'_k\) are mutually dual if the following holds: 
\begin{align}
a'_k = a_{n+1-k}
\end{align}
where \(k \in \{1, \ldots, n\}\). 

The duality of some symbols from the $\Gamma_5$ alphabet in the example above would be $2'=4$ and and $1'=5$. For symbol $3$ is  $3'=3$, which means that it is self-dual (the same holds for $b$ in $\Gamma_3$ alphabet). For $\Gamma_2$ alphabet, two symbols $\beta$ and $\psi$ are mutually dual.

It is common for the ordering of the alphabet to be defined through a valuation sequence (the range of the function $v$) in which two adjacent elements differ by 1, which will allow for easier determination of the distance $\delta(r, s)$ between any two symbols:
\begin{align}
\delta(r, s) = |v(r) - v(s)| \label{eq:del}
\end{align}
where \(r,s \in \Gamma\). 

For example, the distance $\delta(c, a)=2$ from the $\Gamma_3$ 
and the distance $\delta(2, 5)=3$ from the $\Gamma_5$. It is trivial that for the $\Gamma_2$ alphabet if the valuation sequence is \{0,1\} or \{1,2\}, then their mutual distance is $1$, whereas for the same alphabet with a valuation sequence of \{–1,1\}, the distance would be $2$.

\subsection{The concatenation}
Concatenation of two symbols from an $n$-symbol alphabet can be done in $n^2$ ways, which gives rise to a set of symbol pairs called $m$-atoms, which form the basis of the $m$-system. Since $m$-words are formed by concatenation of symbols from an alphabet, an alphabet with $n$ elements is said to be $n$-generator. For example, $\Gamma_3$ is a \textit{three-generator} alphabet, while $\Gamma_5$ is a \textit{five-generator} alphabet. The result of concatenation $x \cdot y $ is often written as $xy$.

Since the set $M_\Gamma$ consists of $m$-words that contain only an \textit{even number} of symbols, $m$-atoms are also the shortest $m$-words. 

For any $m$-word $x\in M_{\Gamma}$ there are $l_x$, the \textit{starting} symbol (\textit{left} char of $x$, $l(x)$) and $r_x$, the \textit{ending} symbol (\textit{right} char of $x$, $r_x$).

It should also be emphasized that \textit{a shell} of $x$ as a function $q:M_{\Gamma}\rightarrow M_{\Gamma}$ such that:
\begin{equation}
q(x):= l(x) r(x) = l_x r_x \label{qlj}
\end{equation}
The shell $q(\,)$ plays an important role in logic, where it represents a logical variable, while in electrical circuits it characterises the circuit’s behaviour at the two frequency extremes, namely at very low ($f=0$) and very high ($f=\infty$) frequencies.

 In electrical networks, the distance between $r_x$ and $l_x$ symbols is particularly important. It is defined analogously to the delta distance \eqref{eq:del} and is called the \textit{defect}:
\begin{align}
\Delta(r_x,l_x) = v(r_x) - v(l_x) \label{eq:def}
\end{align}

\subsection{The Set $M_\Gamma$ } 

The power of \textit{$m$-theory} lies in the fact that an RLC electrical network can be represented by a jorb ($m$-word) that reveals its impedance.  Although it is a formalized mathematical theory, its application holds valuable physical, topological, and logistical interpretations.

In order for working with jorbs as representations of electrical networks to be complete, it is necessary to introduce two additional key concepts: \textit{$\lambda$-length} and \textit{$m$-word compression}.

\subsubsection{$\lambda$-Length }
In mathematics, the length of a word, as an algebraic concept, usually refers to the number of symbols contained in the word. In \textit{jorbology}, the length of a $m$-word characterizes its "waviness". The length of a $m$-word is denoted by $\lambda$ and is equal to the alternated sum of the distances between adjacent letters of the $m$-word, starting with a negative sign.  More precisely, for all $n\in\mathbb{N}$ and for each $a_1, a_2, \ldots, a_{2n} \in \Gamma$:
\begin{equation}
\lambda(a_1 \cdot a_2 \cdot a_3 \cdot \ldots \cdot a_{2n}) = - \delta(a_1, a_2) + \delta(a_2, a_3) - \ldots - \delta(a_{2n-1}, a_{2n}) \label{lam}
\end{equation}

\begin{example}     
For $\Gamma_3 = \{a,b,c\}$ with valuation 
$v(a)=-1$, $v(b)=0$, $v(c)=1$, for the $m$-word 
\[
w = a \cdot b \cdot c \cdot a \cdot b \cdot c ,
\]
we compute
\[
\lambda(w) 
= -\delta(a,b) + \delta(b,c) - \delta(c,a) + \delta(a,b) - \delta(b,c).
\]
Since
\[
\delta(a,b) = |v(a)-v(b)| = 1, \quad
\delta(b,c) = |v(b)-v(c)| = 1, \quad
\delta(c,a) = |v(c)-v(a)| = 2,
\]
it follows that
\[
\lambda(w) = -1 + 1 - 2 + 1 - 1 = -2.
\]
\end{example}

\subsubsection{$m$-word Compression}

Using the concept of $m$-word length $\lambda$, one arrives at the relation of $m$-word compression. The equality of two $m$-words due to the  compression relation will be denoted by “$=_{zip}$”. The relation is first defined for the special case of compressing a four-letter word into a two-letter one, that is:
\begin{equation}
 p \cdot r \cdot s \cdot t =_{zip} p \cdot t \quad \iff \quad  \lambda(p \cdot r\cdot s \cdot t) = \lambda(p \cdot t) \label{saz4}
\end{equation}
for all symbols $p, r, s, t \in \Gamma$. In the general case, the compression of an arbitrary $m$-word $w \in M_\Gamma$ 
is obtained by iteratively applying the elementary rule~\eqref{saz4} to its subwords, 
and by requiring that the relation be preserved under concatenation. 
Formally, for all $m$-words $w_1,w_2,w_3,w_4 \in M_\Gamma$,
\begin{equation}
w_1 =_{zip} w_2 \text { \quad and\quad } w_3 =_{zip} w_4\quad\Longrightarrow\quad  w_1w_3 =_{zip} w_2w_4
\label{saz5}
\end{equation} 
It follows that the compression relation $=_{zip}$ defines an equivalence relation on $M_\Gamma$.

For three consecutive symbols in a $m$-word, where the first and the third are the same, the following compression law also holds:

\begin{proposition} Let $p, r \in \Gamma$ be arbitrary symbols. Then 
\begin{equation}
w_1 \cdot p \cdot r \cdot p \cdot w_2 =_{zip} w_1  \cdot p \cdot w_2 \label{saz3}
\end{equation}
where $w_1,w_2$ are subwords (one of which can be empty) such that $w=w_1 \cdot p \cdot r \cdot p \cdot w_2\in M_\Gamma$.
\end{proposition}
\begin{proof}  By the definition of the $m$-word length $\lambda$,
\[\lambda(w)=\lambda(w_1\cdot p\cdot w_2)+\delta(p,r)-\delta(p,r)=\lambda(w_1\cdot p\cdot w_2),\] so by \eqref{saz4} and \eqref{saz5}, $w_1 \cdot p \cdot r \cdot p \cdot w_2 =_{zip} w_1  \cdot p \cdot w_2$.
\end{proof}

  A more detailed algebraic treatment of the \(m\)-system, particularly its compression operations, is presented in \cite{essert2022vsare}.

\begin{example}\label{ex:jorb}
\leavevmode
For simplicity, instead of two consecutive small letters from alphabet, one uppercase letter is often written.
\begin{itemize}
	\item For $\Gamma_2=\{a, b\}$, $x = baba=_{zip} ba; \hspace{0.2cm} y=aaaaabbb=AAaBb=_{zip}aabb=AB$
	\item For $\Gamma_3=\{a, b, c\}$, $x = bcbaac=bcbAc=_{zip}bAc; \hspace{0.2cm} y=ABC=_{zip}AC$ because $\lambda(a \cdot a \cdot b \cdot b \cdot c \cdot c) = - \delta(a, a) +  \delta(a, b) - \delta(b, b) + \delta(b, c) - \delta(c, c) = - 0 + 1 - 0 + 1 - 0 = 2$ and
	$\lambda(a\cdot a \cdot c\cdot c) = -\delta(a, a) + \delta(a, c) -\delta(c, c)= -0+2-0=2$.  \\
	It is evident that the inverse operation is allowed, i.e., the expansion of an $m$-number: $AC =_{exp} ABC$ (of course, only when dealing with ‘uppercase letters,’ i.e., double identical lowercase ones).
\end{itemize}
\end{example}

\subsubsection{Unary Operators}

Let $D$ be the operator that acts on the entire $m$-number, which transforms a $m$-number into its dual form:
\begin{equation}
(\forall s,t \in \Gamma)\; \;D(st) = s't'\qquad \wedge \qquad(\forall x,y \in M_{\Gamma})\;(x=sty \rightarrow D(x)=s't'D(y)). \label{Dual}
\end{equation}

\begin {figure}[h]
\captionsetup{skip=7pt}
\centering
\setlength{\tabcolsep}{8pt}
\renewcommand{\arraystretch}{1.2}
\begin{tabular}{|c|c|c|c|c|}
\hline
$\cdot$ & 1 & D & E & F \\ \hline
1 & \cellcolor{lightgray}1 & D & E & F \\ \hline
D & D & \cellcolor{lightgray}1 & F & E \\ \hline
E & E & F & \cellcolor{lightgray}1 & D \\ \hline
F & F & E & D & \cellcolor{lightgray}1 \\ \hline
\end{tabular}
\caption{Permutation group}
\label{tab:pgr}
\end{figure}

For word \(x = a_1 \cdot a_2 \cdot \ldots \cdot a_{2k}\) from \(M_\Gamma\), its dual word is \(x' = a'_1 \cdot a'_2 \cdots \ldots \cdot a'_{2k}\), where \(a'_1, a'_2, \ldots, a'_{2k}\) are dual symbols corresponding to \(a_1, a_2, \ldots, a_{2k}\).

Operator E reverses the word, joining the letters from right (the last letter) to left (toward the first)
\begin{equation}
(\forall s,t \in \Gamma)\; \;E(st) = ts \qquad \wedge \qquad(\forall x,y \in M_{\Gamma})\;(x=yst \rightarrow E(x)=tsE(y)).  \label{Eop}
\end{equation}
The F operator is obtained by composing the D and E operators: 
\begin{equation}F(x)=D(E(x))=E(D(x))\label{Fop}
\end{equation}
Together with the identity, these operators under composition form the four-group (Klein group), as shown in the Cayley's table (Fig \ref{tab:pgr}).

\begin{example}
By the action of the operators on an $m$-number, two or four other $m$-numbers are obtained, or the $m$-number remains the same.
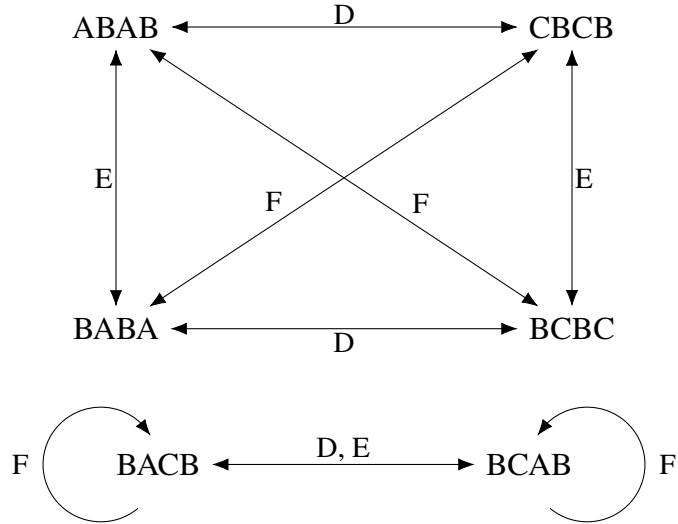
\begin{figure}[ht]
\centering
\begin{tikzpicture}[x=2.0cm,y=2.0cm] 

\tikzset{
  >={Latex[length=2.2mm]},
  lab/.style={font=\small, inner sep=1pt},
}

\node (TL) at (-0.5,2) {ABAB};
\node (TR) at ( 2.5,2) {CBCB};
\node (BL) at (-0.5,0) {BABA};
\node (BR) at ( 2.5,0) {BCBC};

\draw[<->] (TL) -- node[lab,left]  {E} (BL);
\draw[<->] (TR) -- node[lab,right] {E} (BR);

\draw[<->] (TL) -- node[lab,above] {D} (TR);
\draw[<->] (BL) -- node[lab,below] {D} (BR);


\draw[<->] (TL) --
  node[lab,above,pos=0.32,xshift=55pt, yshift=-32pt] {F} (BR);

\draw[<->] (TR) --
  node[lab,above,pos=0.68,yshift=4pt] {F} (BL);
  
\def\rloop{0.38}   
\def\shift{0.6}   

\coordinate (cL) at (-\shift,-0.9);
\coordinate (cR) at ( 2+\shift,-0.9);

\node (L) at ($(cL)+(0:\rloop)$) {BACB};
\node (R) at ($(cR)+(180:\rloop)$) {BCAB};

\draw[<->] (L.east) -- node[lab,above]{D, E} (R.west);

\draw[-{Latex[length=2.2mm]}]
  ($(cL)+(310:\rloop)$) arc[start angle=310, delta angle=-280, radius=\rloop];
\node[lab] at ($(cL)+(-\rloop-0.15,0)$) {F};

\draw[-{Latex[length=2.2mm]}]
  ($(cR)+(230:\rloop)$) arc[start angle=230, delta angle=280, radius=\rloop];
\node[lab] at ($(cR)+(\rloop+0.15,0)$) {F};

\end{tikzpicture}
\caption{Action of group elements D, E, and F}
\label{fig:goran}
\end{figure}

In Fig.\ref{fig:goran} the action of all unary operators is shown. Similar algebraic operations are performed in the classification of network subfamilies (see \cite{phdthesis}) or in the determination of equivalence classes (see \cite{JiangSmith2012}, \cite{morelli2019passive}).
\end{example}

\subsubsection{Relations $\leq_q$ and $\geq_q$}
We  define  relations $\leq_q$ and $\geq_q$ called \textit{q-less or equal } and \textit{q-more or equal } between two $m$-numbers $x,y\in M_\Gamma$ as follows:
\begin{equation}
\begin{aligned}
x \leq_q y &\;\Longleftrightarrow\; l_x \geq_\Gamma l_y \;\wedge\; r_x \leq_\Gamma r_y, \\
x \geq_q y &\;\Longleftrightarrow\; l_x \leq_\Gamma l_y \;\wedge\; r_x \geq_\Gamma r_y.
\end{aligned}
\end{equation}
It is easy to see that these two relations are reflexive and transitive so they are relations of partial-order  on $M_\Gamma$, i.e. $(M_\Gamma,\leq_q )$ and $(M_\Gamma,\geq_q )$ are quasi-ordered sets.

It is easy to notice that, because for every $x\in M_\Gamma$ holds that $l(q(x))=l_x$ and  $r(q(x))=r_x$, where $q$ is defined in (\ref{qlj}), it follows:
\begin{equation}\label{T2}
x\leq_q y\ \Longleftrightarrow\ \ q(x)\leq_q q(y),\ \ \ x\geq_q y\ \Longleftrightarrow\ \ q(x)\geq_q q(y).
\end{equation}
Further, we can define relation of equivalence $=_q$ on $M_\Gamma$ as
\begin{equation}
x=_q y \Longleftrightarrow x\leq_q y \wedge y\leq_q x.
\end{equation}
Now from \eqref{T2} easy follows that
\begin{equation}
x=_q y\Longleftrightarrow q(x)=q(y).
\end{equation}

For example, for some $m$-numbers in $\Gamma = \{0,1,\cdots,9\}$: $5322\leq_q 44332266$, $3326\geq_q 44332262$, $13\leq_q 1433$, $13\geq_q 1433$, $13=_q 1433$

\subsubsection{Relation of Equivalence of $m$-space}
Relation of equivalence $=_q$  generates partition of the set $M_\Gamma$ into $n^2$ classes of equivalence, where $n$ is a number of $m$-space generators (i.e. number of symbols in $\Gamma$). Obtained classes of equivalence are denoted as $[st]$, for all $s,t\in\Gamma$. For example, for $\Gamma=\{a,b,c\}$, we have $ba, baabbcca, BABA, BCBA,... \in [ba]$

With $m$-atom '\textbf{u}' it is easy to see that $[\textbf{u}]=\{x\in M_\Gamma: q(x)=\textbf{u}\}=[l(\textbf{u})r(\textbf{u})]$. For every $\textbf{u}\in M_\Gamma$, $([\textbf{u}],\cdot)$ is an infinite commutative group with concatenation operation '$\cdot$'.
The set of all classes of equivalence (i.e. quotient set) will be denoted by $[M_\Gamma]$, that is $[M_\Gamma]=\{[\textbf{u}]:\textbf{u}\in M_\Gamma\}$.

For every symbol $s\in\Gamma$ we define
\begin{equation}
Lideal(s)=\{x\in M_\Gamma: l_x=s\},\ \ \ Rideal(s)=\{x\in M_\Gamma: r_x=s\}.
\end{equation}
Because for all $s\in\Gamma$ is $Lideal(s)\cdot M_\Gamma=Lideal(s)$ and $M_\Gamma \cdot Rideal(s) =Rideal(s)$\footnote{Here concatenation between two sets is generalization of concatenation operation defined on $M_\Gamma$ in sense that $X \cdot Y =\{x \cdot y: x \in X, y \in Y\}$. For example, $\{aa,bb\}\cdot \{cc,bc,bb\}=\{aacc,aabc,aabb,bbcc,bbbc,bbbb\}$ }, $Rideal(s)$ is right ideal and $Lideal(s)$ is left ideal of semigroup $(M_\Gamma, \cdot)$. Notice that for $s,t\in\Gamma$ we have
\begin{equation}
Lideal(s)\cap Rideal(t)=[st]
\end{equation}
and
\begin{equation}
\{Lideal(s):s\in\Gamma\}\cap \{Rideal(s):s\in\Gamma\}=[M_\Gamma].
\end{equation}

In geometric interpretation, relation of equivalence ais represented by lines  which are intersections of the planes which represent the right and left ideals.

\begin{example} On Figure \ref{fig:re} right ideals are connected by solid lines and left ideals by dashed lines.
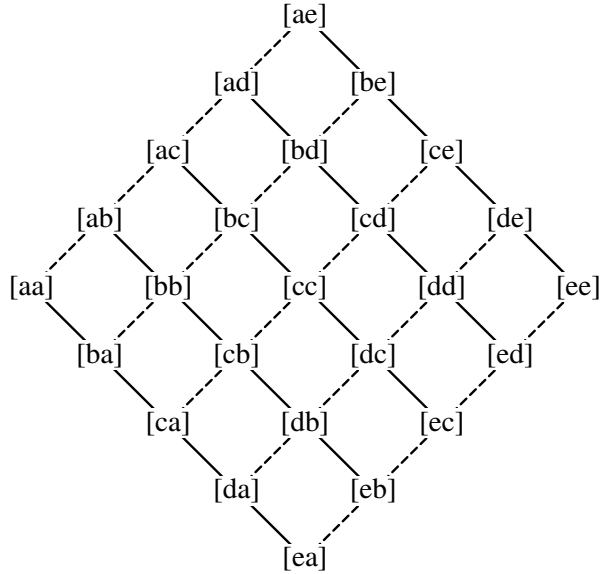
\begin{figure}[ht]
\centering
\begin{tikzpicture}[scale=0.9, line cap=round, line join=round]
  \foreach \i in {1,...,5} {
    \foreach \j in {1,...,5} {
      \coordinate (p\i\j) at ({\i+\j},{\j-\i});
    }
  }

  \foreach \i in {1,...,4} {
    \foreach \j in {1,...,5} {
      \draw[line width=0.9pt] (p\i\j) -- (p\the\numexpr\i+1\relax\j);
    }
  }

  \foreach \i in {1,...,5} {
    \foreach \j in {1,...,4} {
      \draw[line width=0.9pt, dash pattern=on 3pt off 2pt]
        (p\i\j) -- (p\i\the\numexpr\j+1\relax);
    }
  }

  \foreach \i in {1,...,5} {
    \foreach \j in {1,...,5} {
      \node[font=\small, fill=white, inner sep=1pt]
        at (p\i\j) {[{\Ltr{\i}}{\Ltr{\j}}]};
    }
  }
\end{tikzpicture}
\caption{Hasse diagram for $([M_\Gamma],\le_q)$ with $\Gamma=\{a,b,c,d,e\}$.}
\label{fig:re}
\end{figure}
\end{example}

\section{Jorb as an RLC Network Representation}
The set of jorbs in the $M_\Gamma$ space can be viewed as a set of impedance representations of individual networks, starting from the simplest electrical elements:
\begin{itemize}
	\item A capacitor with capacitance \textit{C} is represented by the $m$-atom \textit{'aa'}, which we will denote as 'A', 
	\item A resistor with resistance \textit{R} will be represented by \textit{'bb'} ('B'),
	\item An inductor with inductance \textit{L} will be represented by 'cc' ('C').
\end{itemize}
It is evident that in this case the three-generator alphabet is sufficient, i.e. $\Gamma_3 = \{a, b, c\}$ is used, with a valuation sequence of $\{1, 2, 3\}$ or alternatively $\{-1, 0, 1\}$.

\subsection{Series and Parallel Connections}

The first problem is how to represent an electrical network consisting of two branches that are connected in series (or in parallel)\footnote{For the series connection of two jorbs, we will use the symbol '$\eqcirc$' (or 's'), and for the parallel connection, the symbol '$\pl$' (or 'p'). }, each containing, for example, two different elements: a capacitor and a resistor, i.e., represented by the $m$-atoms $A$ and $B$. The resulting series impedance should not depend on the order of connection, i.e., whether $A$ is connected in series with $B$ or vice versa. To achieve this consistently in the algebra of $m$-atoms, M.~Šare in \cite{Sare2000} (p.~40) proposed the following general formula for the series connection of two jorbs:

\begin{equation}
x \eqcirc  y :=
(l_x \downarrow l_y) \cdot (l_x \downarrow l_y) \cdot x  \cdot (r_x \uparrow l_y)\cdot (r_x \downarrow l_y) \cdot y \cdot (r_x  \uparrow r_y) \cdot (r_x \uparrow r_y) \label{eq:ser}
\end{equation}
By substituting the simplest test case — $A$ in series with $B$ and $B$ in series with $A$ — into the formula \eqref{eq:ser}, we obtain: 
\begin{align*}
	  aa \eqcirc bb &= (a \downarrow b) \cdot (a \downarrow b) \cdot
	  aa  \cdot (a \uparrow b)\cdot (a \downarrow b) \cdot bb \cdot (a  \uparrow b) \cdot (a \uparrow b)\\
	  &= a  \cdot a  \cdot aa  \cdot b\cdot a  \cdot bb \cdot b \cdot b = aaaababbbb = AAbaBB =_{zip} aabb = AB\\
    bb \eqcirc aa &= (b \downarrow a) \cdot (b \downarrow a) \cdot
   bb  \cdot (b \uparrow a)\cdot (b \downarrow a) \cdot aa \cdot (b  \uparrow a) \cdot (b \uparrow a)\\
    &= a  \cdot a  \cdot bb  \cdot b\cdot a  \cdot aa \cdot b \cdot b = aabbbaaabb = ABbAaB =_{zip} aabb = AB
\end{align*}	
At first glance, jorb 'AAbaBB' is not equal to jorb 'ABbAaB'. However, after the application of the \textit{zip} function, it becomes evident that these two jorbs are equivalent.

Analogously, the same holds for the parallel connection (two edges connect at two points (nodes)) of two electrical branches: 
\begin{equation}
x \pl y := (l_x \uparrow l_y) \cdot (l_x \uparrow l_y) \cdot
x  \cdot (r_x \downarrow l_y)\cdot (r_x \uparrow l_y) \cdot y \cdot (r_x  \downarrow r_y) \cdot (r_x \downarrow r_y) \label{eq:par}
\end{equation}
In the same way, it can be verified that $ A \pl B = B \pl A = BA$.

So far, in the context of RLC networks, we have considered only three elements represented by jorbs. To elucidate the physical meaning of the remaining $M‑atoms$ in the alphabet $\Gamma_3$, we proceed by determining the shell of a jorb:
 \[q(AB) = l(AB) \cdot r(AB) = ab\] and similarly \[q(BA)=l(BA) \cdot r(BA) = ba \] which corresponds to:
\begin{enumerate}
\item \textbf{ab}: an \textit{open circuit} (very high impedance, no current flow for $f=0$) in the case of a series connection of a capacitor and a resistor under DC conditions;
\item \textbf{ba}: a \textit{short circuit} (low impedance, but not $0$ - primarily determined by the resistor) in the case of a parallel connection of a resistor and a capacitor.
\end{enumerate}
Indeed, connecting the branch '$ab$' in parallel with any $m$-atom over $\Gamma_2 = \{a, b\}$ leaves the $m$-atom (after compression) unchanged, while connecting '$ba$' in parallel results  in '$ba$'. For example,  $A \pl ab = AAAabA =_{zip} A$ and $B \pl ab = BBabaBb =_{zip} B$, while $A \pl ba = BAaBAa =_{zip} ba$ and $B \pl ab = BBBbAa =_{zip} ba$. The analogous statement holds for the impedance shells of an electrical network composed of multiple branches to which ‘ab’ or ‘ba’ is added in parallel. For example, $q(ABAB \pl ba) = ba$, while $q(ABAB \pl ab) = ab$.

In an analogous fashion one can treat the remaining $M$‑atoms ($ac$, $ca$, $bc$, $cb$); however, their behaviour assumes a different—sometimes dual—form.

\begin{example}\label{ex:elcirc}
\leavevmode\\
As an example, consider a typical electrical circuit consisting of two resistors, two capacitors, and one inductor, as shown in the figure ~\ref{fig:ec}. For $m$-numbers, the actual values of the electrical properties (resistance, capacitance, or inductance) are not relevant—only their type matters, while the structure is determined by the topology, i.e., the way the components are connected.

An intrinsic part of every electrical network is a \emph{graph} whose
branches represent electrical elements. By an \emph{$m$--graph} of a CRL network we mean the $m$--graph obtained by an
isomorphic mapping of the electrical network such that a capacitive branch
always corresponds to an $A$‑branch, a resistive branch to a $B$‑branch, and an
inductive branch to a $C$‑branch.

The corresponding topological graph of the same circuit is shown in Figure~\ref{fig:ec} (right), where the electrical elements are replaced by their equivalent $m$-atoms. 
\begin{figure}[ht]
\centering
\begin{minipage}{0.67\textwidth}
\begin{circuitikz}[american]
  \coordinate (N1) at (0,0);      
  \coordinate (N2) at (0,-3);     
  \coordinate (N3) at (4,0);      
  \coordinate (N4) at (4,-3);

  \draw
    (N1) node[left]{\large 1}
          to[R={$R_1 = 4\,\Omega$},*-*] (N3)
    (N2) node[left]{\large 2}
          to[R={$R_2 = 5\,\Omega$},*-*] (N4)
    (N1) to[C={$C_1 = 2\,\mu\text{F}$}] (N2)
    (N3) to[L={$L_1 = 3\,\text{mH}$}]   (N4);

  \draw
    (N3) to[short,-] ++(3,0)
          to[C={$C_2 = 20\,\text{nF}$}] ++(0,-3)
          to[short,-] (N4);

  \node[above] at (N3) {\large 3};
  \node[below] at (N4) {\large 4};
\end{circuitikz}
\end{minipage}
\begin{minipage}{0.30\textwidth}
\begin{circuitikz}[american]
  \coordinate (N1) at (0,0);
  \coordinate (N2) at (0,-3);
  \coordinate (N3) at (2,0);
  \coordinate (N4) at (2,-3);

  \draw
    (N1) node[left]{\large 1}
          to[short,l=B,*-*] (N3)
    (N2) node[left]{\large 2}
          to[short,l=B,*-*] (N4)
    (N1) to[short,l=A] (N2)
    (N3) to[short,l=C] (N4);

  \draw
    (N3) to[short,-] ++(1,0)
          to[short,l=A] ++(0,-3)
          to[short,-] (N4);

  \node[above] at (N3) {\large 3};
  \node[below] at (N4) {\large 4};
\end{circuitikz}
\end{minipage}
\caption{Electrical circuit with the corresponding $m$-graph}
\label{fig:ec}
\end{figure}
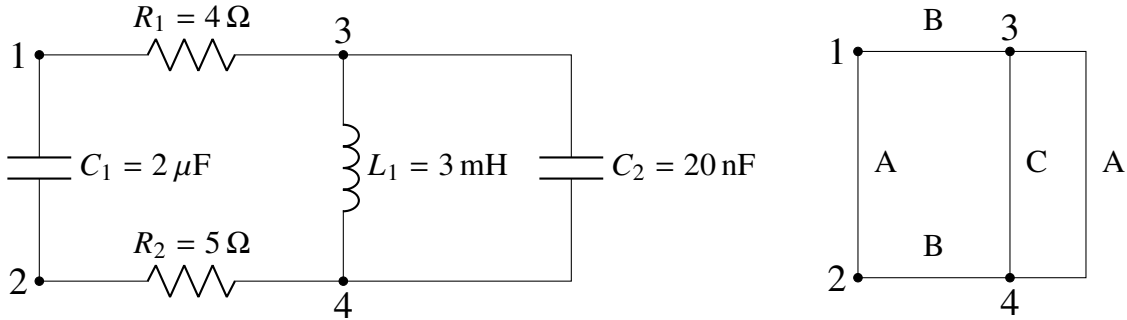

Using the rules defined by the equations (\ref{eq:ser}) and (\ref{eq:par}), the corresponding jorb between two nodes (i.e., a one-port network) can be easily computed. In this example, nodes $1$ and $2$ are considered in the first case, and nodes $1$ and $4$ in the second, and their respective jorbs are calculated. It is evident that the impedance between the selected nodes is not the same.

$Z_{(1,2)}= A \p (B \s (C \p A) \s B) = BCABA$

$Z_{(1,4)}= (A \s B) \p (B \s (C \p A)) = BACAB$

If an additional branch 'ca' (representing a "short circuit" for those components) is added in parallel to the parallel connection of A and C, then the total impedance of these parallel branches between nodes 3 and 4 would be precisely 'ca' i.e. $(C\; p\; A)\; p\; ca = CA\; p\; ca = ca$. As a result, the overall impedance—and thus the jorb—between nodes 1 and 4 would also change.

$Z_{(1,4)}= (A \s B) \p (B \s (C \p A \p ca)) = cbA$
\end{example}

\subsection{S-Core  and  P-Core}

There is $m$-atoms $\omega \alpha$ called \textit{s-zero} and  $m$-atoms $\alpha \omega $  called \textit{p-zero} which\footnote{We recall: $\alpha$ is the first element of the $\Gamma$ alphabet and $\omega$ is the last one.}, when connected in series or in parallel with any jorb, do not alter the value of the jorb they are connected to:
\begin{equation}
(\forall x \in M_\Gamma)\, (x \eqcirc \omega \alpha = x) \qquad \& \quad (x \pl \alpha
 \omega  = x) \label{spzero}
\end{equation}

The application of equation (12), in the case when it is not possible to create a node in the star connection for the delta–star transformation (fig.  \ref{fig:delta}), to the branches whose jorbs cannot reproduce the original branch’s jorb through series connection, often uses an $s-zero$ to split the branch into two parts, thereby resolving the problem.

The \textit{s-core} ($J_s$), respectively the \textit{p-core} ($J_p$), of a jorb is obtained by attaching an \textit{s-zero} or \textit{p-zero}, respectively, to both ends of the jorb. In $m$-topology, the \textit{s-core} of a jorb represents a fundamental topological \textit{invariant} of $m$-graphs, where it was first identified.
\begin{equation}
(\forall x \in M_\Gamma)\, (J_s(x) = \omega \alpha \cdot x \cdot \omega \alpha) \quad \& \quad (J_p(x) = \alpha \omega  \cdot x \cdot \alpha \omega) \label{Js}
\end{equation}

It follows straightforward that for every $x\in M_\Gamma$
\begin{equation}
J_s(x)=(x\eqcirc \alpha \omega)\pl \omega\alpha,\quad J_p(x)=(x \pl\omega\alpha)\eqcirc \alpha \omega.
\end{equation}

\begin{example} 

Let us verify that the $J_s$ indeed renders the jorb (i.e. the electrical network graph) invariant—meaning it yields the same jorb regardless of the pair of nodes from which it is observed. From the previous example (ex. \ref{ex:elcirc}), we should have:
\[J_s(BCABA)=J_s(BACAB)\]
Let's compute:
\begin{equation*}
\begin{aligned}J_s(BCABA)&\;=caBCABAca=cabbccaabbaaca=cbbabbccbbaabbaaca=_{zip}CABA\\
J_s(BACAB)&\;=caBACABca=cbbabbaabbccbbaabbcbba=_{zip}CABA
\end{aligned}
\end{equation*}
Recall that by the compression/expansion law: $ca=cbba$ and $ac=abbc$. Also, note that when written in uppercase letters, both expressions could be represented as $CBABA$, too. This confirms that the \textit{s-core} preserves the structure of the jorb, making it a topological invariant.
\end{example}
\subsection{S-Shell and P-Shell}

In every $m$-system there is one and only one $m$-number with property that if we add it by serial connection to s-core of $x$ we obtain $x$. This $m$-number is called s-shell of $x$ and it is denoted by $q_s(x)$. Similarly, p-shell $q_p(x)$ of $x$ is a unique $m$-number which added to p-core of $x$ by parallel connection gives $x$.  They are given by
\begin{equation}
\label{spshell}
\begin{aligned}
q_s(x) &= l_x(l_x\downarrow r_x)(l_x\uparrow r_x)r_x,\\
q_p(x) &= l_x(l_x\uparrow r_x)(l_x\downarrow r_x)r_x.
\end{aligned}
\end{equation}
It is easy to notice that if two $m$-numbers have the same shell, then they also have the  same s-shell as well as the same p-shell, i.e. \[q(x)=q(y)\Longrightarrow q_s(x)=q_s(y) \wedge q_p(x)=q_p(y).\]

Moreover, for the operator $I(x)=l_x x r_x$ ($x\in M_\Gamma$) we have
\begin{equation}\label{Ix}
x\eqcirc I(x)=q_p(x),\qquad x\pl I(x)=q_s(x),
\end{equation}
and
\begin{equation}\label{JsJp}
J_s(x)\eqcirc q_s(x)=x,\qquad J_p(x)\pl q_p(x)=x.
\end{equation}

Finally, the sets
\begin{equation}\label{Omega}
\Omega_s=\{q_s(x):x\in M_\Gamma\},\qquad 
\Omega_p=\{q_p(x):x\in M_\Gamma\},
\end{equation}
are called the \emph{s-basis} and \emph{p-basis} of the $m$-space.

s-base is the biggest subset of $M_\Gamma$ which is closed under $\eqcirc$ operation. $ (\Omega_s,\leq_q)$ is a partially ordered set. On the whole $M_\Gamma$, relation $\leq_q$ is transitive and reflexive so it is also on $\Omega_s$. But on the whole $M_\Gamma$ it is not antisymmetric, while on $\Omega_s$ it is because $q_s(x)\leq_q q_s(y)$ and $q_s(y)\leq_q q_s(x)$ imply $l_x=l(y)$ and $r_x=r(y)$, which means that $q_s(x)=q_s(y)$.

Further, because $(M_\Gamma, \eqcirc)$ is commutative semigroup and each  element in $\Omega_s$ is idempotent i.e. $q_s(x)\eqcirc q_s(x)=q_s(x)$, it means that  $(\Omega_s,\eqcirc)$ is a commutative band, i.e. upper semilattice with respect to $\leq_q$ relation.

On the other hand, $(\Omega_p,\pl)$ is a commutative band and  lower semilattice with respect to $\leq_q$ relation and upper semilattice with respect to $\geq_q$. relation.

\begin{figure}[ht]
\centering

\begin{minipage}[t]{0.47\textwidth}
\centering
\textit{(a) $(\Omega_s,\le_q)$}\\[0.3em]
\begin{tikzpicture}[scale=0.9, line cap=round, line join=round]
  \foreach \i in {1,...,5} {
    \foreach \j in {1,...,5} {
      \coordinate (s\i\j) at ({\i+\j},{\j-\i});
    }
  }
  \foreach \i in {1,...,4} {
    \foreach \j in {1,...,5} {
      \draw (s\i\j) -- (s\the\numexpr\i+1\relax\j);
    }
  }
  \foreach \i in {1,...,5} {
    \foreach \j in {1,...,4} {
      \draw (s\i\j) -- (s\i\the\numexpr\j+1\relax);
    }
  }
  \foreach \i in {1,...,5} {
    \foreach \j in {1,...,5} {
      \ifnum\i<\j
        \node[font=\small, fill=white, inner sep=1pt] at (s\i\j) {\Up{\i}\Up{\j}};
      \else\ifnum\i=\j
        \node[font=\small, fill=white, inner sep=1pt] at (s\i\j) {\Up{\i}};
      \else
        \node[font=\small, fill=white, inner sep=1pt] at (s\i\j) {\Lo{\i}\Lo{\j}};
      \fi\fi
    }
  }
\end{tikzpicture}
\end{minipage}
\hfill
\begin{minipage}[t]{0.47\textwidth}
\centering
\textit{(b) $(\Omega_p,\ge_q)$}\\[0.3em]
\begin{tikzpicture}[scale=0.9, line cap=round, line join=round]
  \foreach \i in {1,...,5} {
    \foreach \j in {1,...,5} {
      \coordinate (p\i\j) at ({\i+\j},{\j-\i});
    }
  }
  \foreach \i in {1,...,4} {
    \foreach \j in {1,...,5} {
      \draw (p\i\j) -- (p\the\numexpr\i+1\relax\j);
    }
  }
  \foreach \i in {1,...,5} {
    \foreach \j in {1,...,4} {
      \draw (p\i\j) -- (p\i\the\numexpr\j+1\relax);
    }
  }
  \foreach \i in {1,...,5} {
    \foreach \j in {1,...,5} {
      \ifnum\i<\j
        \node[font=\small, fill=white, inner sep=1pt] at (p\i\j) {\Up{\j}\Up{\i}};
      \else\ifnum\i=\j
        \node[font=\small, fill=white, inner sep=1pt] at (p\i\j) {\Up{\i}};
      \else
        \node[font=\small, fill=white, inner sep=1pt] at (p\i\j) {\Lo{\j}\Lo{\i}};
      \fi\fi
    }
  }
\end{tikzpicture}
\end{minipage}

\caption{Hasse diagrams for $(\Omega_s,\le_q)$ i $(\Omega_p,\ge_q)$.}
\label{fig:omega-hasse}
\end{figure}
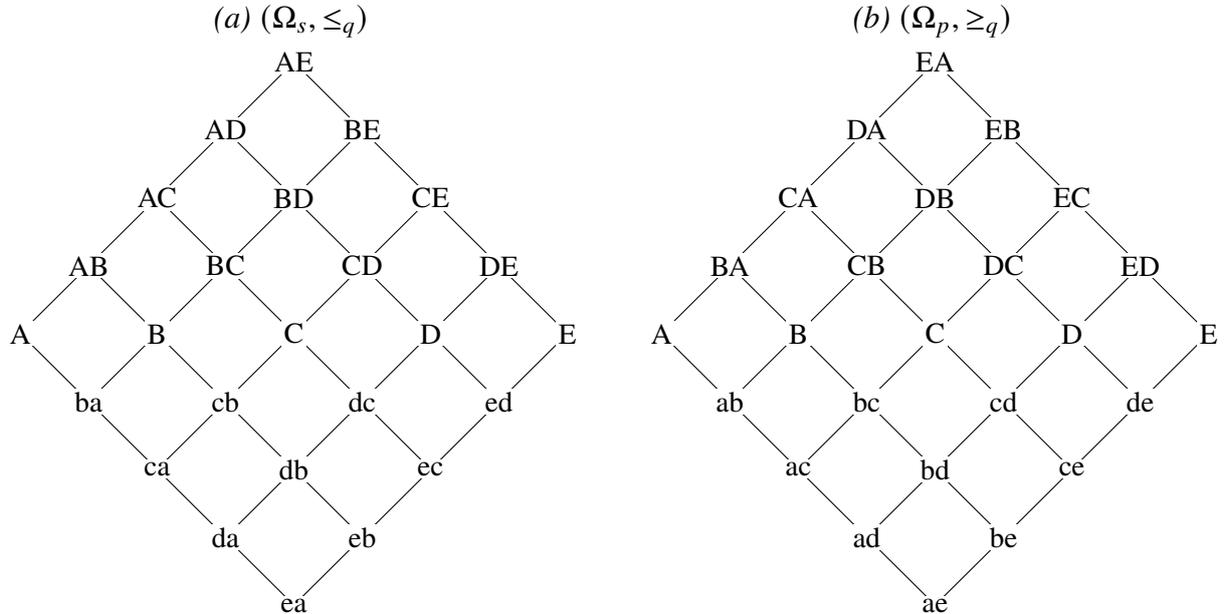

\begin{example} 

Let $\Gamma=\{a,b,c,d,e\}$. Then
$$\Omega_s=\{A,AB,AC,AD,AE,ba,B,BC,BD,BE,ca,cb,C,CD,CE,da,db,dc,D,DE,ea,eb,ec,ed,E\}$$ and
$$\Omega_p=\{A,ab,ac,ad,ae,BA,B,bc,bd,be,CA,CB,C,cd,ce,DA,DB,DC,D,de,EA,EB,EC,ED,E\}.$$
Hasse diagrams of $(\Omega_s,\leq_q)$ and $(\Omega_p,\geq_q)$
are given on Figure \ref{fig:omega-hasse}.

\end{example}

\subsection{$m$-Topology}

An \(m\)-graph $G_\Gamma$ is an ordered pair \((V,E)\), where \(V\) is any nonempty set whose elements are called the vertices of the \(m\)-graph.
The second component \(E\) is a relation whose domain is the set of unordered pairs of vertices and whose range is the set of jorbs \(M_\Gamma\). The \(m\)-topology is also a generalization of the standard topology of a linear geometric graph, since the latter can be  obtained as a special case of the \(m\)-topology.

The edges of an \(m\)-graph are ordered pairs \((\{j,k\},x)\in E\), where \(j,k\) are the vertices of the edge and \(x\) is a jorb, called the name of the edge.  In \(m\)-topology, the function that assigns to each edge (or two-pole network) 
its jorb name is called the \emph{theta function}, denoted by~\(\vartheta\).  The name of an edge is also understood as the name of a binary relation on the vertex set.
For every chain \(\Gamma\), for every set of jorbs \(M_\Gamma\), and for every nonempty set \(X\), the pair \((X,Y)\) is an \(m\)-graph over \(M_\Gamma\) if and only if
\begin{equation}
  Y \;\subseteq\; \bigl\{\,(\{j,k\},x)\;:\; j,k\in X \ \wedge\ x\in M_\Gamma \,\bigr\}.
\end{equation}
\noindent

In this sense, if \(j(y)k\) is the \(m\)-two-pole (one-port network) obtained from \(j(x)k\) by replacing two series edges in \(x\) with a single edge whose name equals the series sum of the names of the two replaced edges, then \(\vartheta(j(y)k)=\vartheta(j(x)k)\).
Similarly, the theta of an \(m\)-two-pole remains unchanged if two parallel edges are replaced by a single edge whose name equals the parallel sum of the names of the replaced edges.
By the procedure described, every nonseparable series--parallel \(m\)-two-pole can be reduced to a single edge whose name equals the theta of the two-pole.
Through \(\vartheta\), all unary functions of the \(m\)-system are carried over to sets of \(m\)-two-poles; in this sense one also speaks of the shell, core, length, etc., of an \(m\)-two-pole.

Here we list two fundamental transformations of the \(m\)-topology:

\begin{itemize}
\item \textbf{Law of modular distributivity.}

For the configurations \(G_1\) and \(G_2\) in Fig.~\ref{fig:lod}, with edge names \(u,v,w\in M_\Gamma\) labeled as shown, we have
\begin{equation}
  \vartheta(G_1) = \vartheta(G_2)
  \;\Longleftrightarrow\;
  u \,\ge_q\, w
  \;\Longleftrightarrow\;
  \bigl(\,\ell(u)\,\le_{\Gamma}\,\ell(w)\ \wedge\ r(u)\,\ge_{\Gamma}\,r(w)\,\bigr).
  \label{eq:distr}
\end{equation}
On this rule one can base procedures that, in many cases (in particular for series--parallel two-poles and with admissible $\Delta\!\leftrightarrow\!Y$ transformations where needed), reduce an \(m\)-two-pole while keeping its \(\vartheta\)-value unchanged.

\tikzset{
  dot/.style={circle,fill=black,inner sep=1.5pt},
  numlabel/.style={font=\large},
  seglabel/.style={font=\Large},
}

\newcommand{\DiagramX}{%
  \begin{tikzpicture}[baseline=(current bounding box.south)]
    \coordinate (A1) at (0,0);
    \coordinate (A2) at (2,0);
    \coordinate (A3) at (4,0);
    \draw (A1) -- (A3);
    \draw ($(A1)+(0.08,0)$) .. controls +(1.2,1.6) and +(-1.2,1.6) ..
          ($(A3)+(-0.08,0)$)
          node[pos=.50,above=6pt,seglabel] {$u$};
    \node[seglabel] at ($(A1)!0.55!(A2)+(0,-0.35)$) {$v$};
    \node[seglabel] at ($(A2)!0.55!(A3)+(0,-0.35)$) {$w$};
    \node[dot] at (A1) {}; \node[numlabel] at ($(A1)+(-0.10,-0.38)$) {1};
    \node[dot] at (A2) {}; \node[numlabel] at ($(A2)+(0,0.38)$) {2};
    \node[dot] at (A3) {}; \node[numlabel] at ($(A3)+(0.10,-0.38)$) {3};
  \end{tikzpicture}%
}

\newcommand{\DiagramY}{%
  \begin{tikzpicture}[baseline=(current bounding box.south)]
    \coordinate (B1) at (0,0);
    \coordinate (B2) at (2,0);
    \coordinate (B3) at (4,0);
    \draw (B1) -- (B3);
    \node[seglabel] at ($(B1)!0.55!(B2)+(0,-0.35)$) {$v$};
    \node[seglabel] at ($(B2)!0.55!(B3)+(0,-0.35)$) {$w$};
    \draw ($(B1)+(0.08,0)$) .. controls +(0.9,1.3) and +(-0.9,1.3) ..
          ($(B2)+(-0.08,0)$)
          node[pos=.50,above=6pt,seglabel] {$u$};
    \node[dot] at (B1) {}; \node[numlabel] at ($(B1)+(-0.10,-0.38)$) {1};
    \node[dot] at (B2) {}; \node[numlabel] at ($(B2)+(0,0.38)$) {2};
    \node[dot] at (B3) {}; \node[numlabel] at ($(B3)+(0.10,-0.38)$) {3};
  \end{tikzpicture}%
}

\begin{figure}[h!]
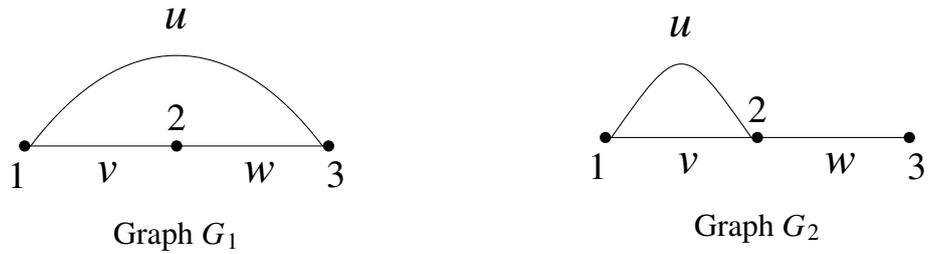

  \centering
  \begin{minipage}{0.48\textwidth}
    \centering
    \DiagramX \\[0.5ex]
    Graph \(G_1\)
  \end{minipage}\hspace{0pt}%
  \begin{minipage}{0.48\textwidth}
    \centering
    \DiagramY \\[0.5ex]
    Graph \(G_2\)
  \end{minipage}
  \caption{The law of modular distributivity}
  \label{fig:lod}
\end{figure}

\item \textbf{Triangle-to-star ($\Delta$--$Y$) transformation.}

An edge subdivision replaces a single edge by a series connection of two edges, with the name of the original edge equal to the series sum of the new names.  
By subdividing edges and applying the law of modular distributivity (and, where applicable, $\Delta\!\leftrightarrow\!Y$ transforms), many two-poles can be transformed into a series--parallel two-pole.

\begin{figure}[h!]
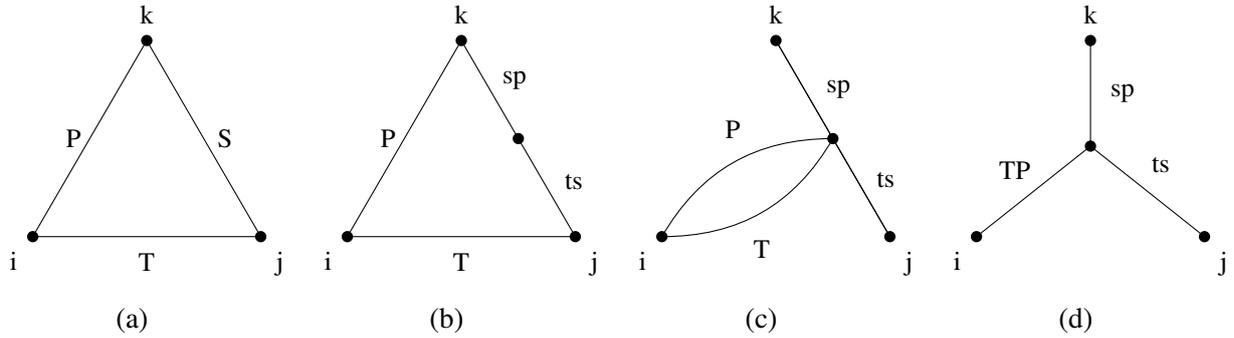

  \centering
  \begin{subfigure}{0.22\linewidth}
    \centering
    \DiagA
    \caption{}
  \end{subfigure}\hfill
  \begin{subfigure}{0.22\linewidth}
    \centering
    \DiagB
    \caption{}
  \end{subfigure}\hfill
  \begin{subfigure}{0.22\linewidth}
    \centering
    \DiagC
    \caption{}
  \end{subfigure}\hfill
  \begin{subfigure}{0.22\linewidth}
    \centering
    \DiagD
    \caption{}
  \end{subfigure}
  \caption{$\Delta$--$Y$ transformation}
  \label{fig:delta}
\end{figure}

The procedure for converting a $\Delta$ graph (Fig.~\ref{fig:delta}(a)) into a $Y$ graph is as follows.  
Let \(V = \{i,j,k\}\), and let \(p,s,t\in\Gamma\) satisfy \(p \le_\Gamma s \le_\Gamma t\).

\begin{enumerate}
  \item Split one branch into two parts by introducing a new, unnamed node and assign \(m\)-atoms to the new branches such that their series connection equals the original \(m\)-number of that branch (Fig.~\ref{fig:delta}(b)).  
        If this is not possible, use the \(s\)-zero, which in series with any \(m\)-number yields that \(m\)-number.
  \item By the law of modular distributivity (Eq.~\ref{eq:distr}), the remaining branches with their nodes may ``slide'' (Fig.~\ref{fig:lod}) along the newly formed branches toward the unnamed node (Fig.~\ref{fig:delta}(c)).
  \item Finally, the branch of the star obtained by the parallel connection of two branches receives the jorb name of that connection (TP in Fig.~\ref{fig:delta}(d)).
\end{enumerate}
\end{itemize}

\noindent
\begin{example}

This example (Fig.~\ref{fig:k33}) shows the transformation of the double Wheatstone bridge
(the \emph{bridge ladder}—two Wheatstone bridges in series).  
From a graph perspective, it is a two-terminal network (terminals 1 and 4) that is not series–parallel because it contains \(K_4\) as a minor.

\begin{figure}[h!]

\begin{tikzpicture}[
  edgelabel/.style={fill=white,inner sep=1pt}
]

\coordinate (x1) at (0,0);
\coordinate (n2) at (1.4,1.2);
\coordinate (n3) at (3.4,1.2);
\coordinate (n4) at (4.8,0);
\coordinate (n5) at (3.4,-1.2);
\coordinate (n6) at (1.4,-1.2);

\node[draw=none,fill=none] at (-0.70,-0.015) {$X=1$};

\fill (x1) circle (1.2pt);
\fill (n2) circle (1.2pt);    \node[above=2pt of n2] {2};
\fill (n3) circle (1.2pt);    \node[above=2pt of n3] {3};
\fill (n4) circle (1.2pt);    \node[right=2pt of n4] {4};
\fill (n5) circle (1.2pt);    \node[below=2pt of n5] {5};
\fill (n6) circle (1.2pt);    \node[below=2pt of n6] {6};

\draw (x1) -- node[edgelabel,sloped,above]{B} (n2);
\draw (n2) -- node[edgelabel,above]{A} (n3);           
\draw (n3) -- node[edgelabel,sloped,above]{B} (n4);
\draw (n4) -- node[edgelabel,sloped,below]{A} (n5);
\draw (n6) -- node[edgelabel,below]{B} (n5);           
\draw (x1) -- node[edgelabel,sloped,below]{A} (n6);

\draw (x1) -- node[edgelabel,pos=0.24,below]{B} (n4);  

\draw (n2) -- node[edgelabel,pos=0.75,below left]{A} (n5);

\draw (n3) -- node[edgelabel,sloped,below,pos=0.10,rotate=95,allow upside down,xshift=7pt]{B} (n6);

\coordinate (y1) at (7.0,0);
\coordinate (m2) at (8.5,1.2);
\coordinate (m3) at (10.5,1.2);   
\coordinate (m4) at (11.8,0);
\coordinate (m5) at (10.4,-1.2);  
\coordinate (m6) at (8.6,-1.2);

\node[draw=none,fill=none] at (6.30,-0.015) {$Y=1$};

\fill (y1) circle (1.2pt);
\fill (m2) circle (1.2pt);   \node[above=2pt of m2] {2};
\fill (m3) circle (1.2pt);   \node[above=2pt of m3] {3};
\fill (m4) circle (1.2pt);   \node[right=2pt of m4] {4};
\fill (m5) circle (1.2pt);   \node[below=2pt of m5] {5};
\fill (m6) circle (1.2pt);   \node[below=2pt of m6] {6};

\draw (y1) -- node[edgelabel,sloped,above]{B} (m2);
\draw (y1) -- node[edgelabel,sloped,below]{A} (m6);

\draw (m2) to[bend left=28]  node[edgelabel,above]{A} (m3);
\draw (m2) to[bend right=28] node[edgelabel,below]{B} (m3);

\draw (m3) to[bend left=26]  node[edgelabel,above,pos=0.55,yshift=4pt]{B} (m4);
\draw (m3) to[bend right=26] node[edgelabel,below,pos=0.25,yshift=-7pt]{A} (m4);

\draw (m4) to[bend left=36]  node[edgelabel,right,pos=0.55,xshift=5pt]{A} (m5);
\draw (m4) to[bend right=26] node[edgelabel,right,pos=0.50,xshift=-15pt]{B} (m5);

\draw (m6) -- node[edgelabel,below]{B} (m5);

\end{tikzpicture}
\caption{The transformation of double Wheatstone bridge graph}
\label{fig:k33}
\end{figure}
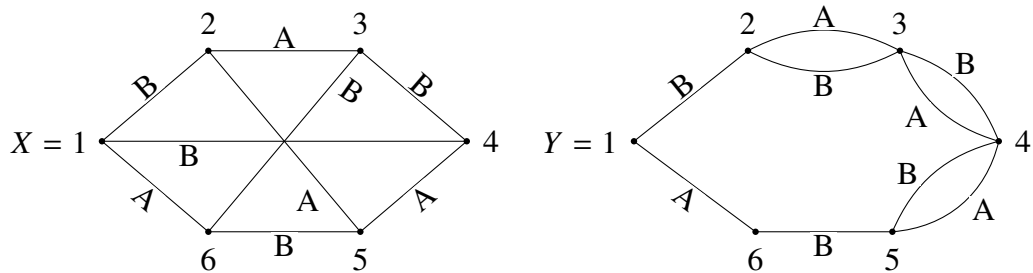

It is easy to see that branch \(B\) \((1\!-\!4)\) can be ``slid'' from node 1 to node 2
and simultaneously from node 4 to node 3.  
Similarly, branch \(A\) \((2\!-\!5)\) can be made parallel to branch \(B\) \((3\!-\!4)\),
and branch \(B\) \((3\!-\!6)\) parallel to branch \(A\) \((4\!-\!5)\),
which enables a series–parallel reduction and computation of the equivalent
impedance between nodes 1 and 4 (or any other pair of nodes).

After the transformation, one can compute the impedance (jorb) between nodes 1 and 4 (or between any other pair of nodes):
\[
\vartheta(1,4) = \vartheta(1',4') = (B \eqcirc BA \eqcirc BA) \eqcirc (A \eqcirc B \eqcirc BA)
               = BABAB \eqcirc ABAB = BABABABAB.
\]

The next example (Fig.~\ref{fig:k4}) requires the $\Delta\!\leftrightarrow\!Y$ transformation
(the modular distribution law alone is not sufficient) to obtain a
series--parallel (SP) graph.  
The example is the complete graph \(K_4\), which we may call a Wheatstone bridge with parallel branches.  
It is planar, so such transformations are possible (unlike for \(K_5\) or \(K_{3,3}\)).

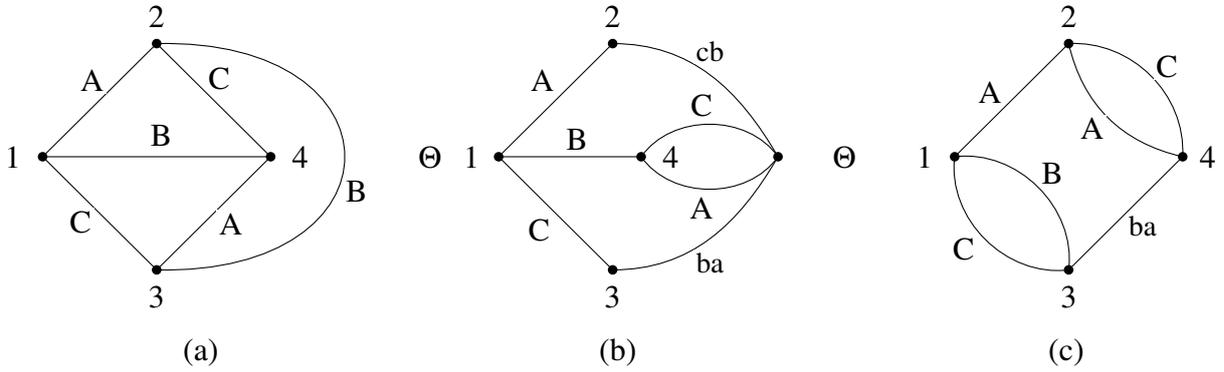
\begin{figure}[ht]
\centering
\begin{tikzpicture}[
  scale=1.5,
  edgelabel/.style={fill=white,inner sep=1pt}
]
\def\H{2.0}      
\def\DX{4.0}     
\def\LBL{3mm}   

\begin{scope}[shift={(0,0)},local bounding box=G1]
  \coordinate (A1) at (0,0);
  \coordinate (A2) at (1, \H/2);
  \coordinate (A3) at (1,-\H/2);
  \coordinate (A4) at (2.0,0);

  \fill (A1) circle (1.2pt); \node[left=4pt of A1] {1};
  \fill (A2) circle (1.2pt); \node at ($(A2)+(0,0.24)$) {2};
  \fill (A3) circle (1.2pt); \node at ($(A3)+(0,-0.24)$) {3};
  \fill (A4) circle (1.2pt); \node at ($(A4)+(0.26,0)$) {4};

  \draw (A1) -- node[edgelabel,above,pos=.55,xshift=-5pt]{A} (A2);
  \draw (A1) -- node[edgelabel,below,pos=.45,xshift=-5pt]{C} (A3);
  \draw (A1) -- node[edgelabel,above,pos=.52, yshift=3pt]{B} (A4);

  \draw (A2) -- node[edgelabel,above,pos=.55,yshift=5pt]{C} (A4);
  \draw (A4) -- node[edgelabel,below,pos=.35,yshift=-5pt]{A} (A3);

  \draw (A2) .. controls ($(A4)+(1.2,1.05)$) and ($(A4)+(1.2,-1.05)$) ..
        node[edgelabel,right,pos=.60,xshift=2pt]{B} (A3);
\end{scope}

\begin{scope}[shift={(\DX,0)},local bounding box=G2]
  \coordinate (B1) at (0,0);
  \coordinate (B2) at (1, \H/2);
  \coordinate (B3) at (1,-\H/2);
  \coordinate (B4) at (1.25,0);
  \coordinate (BR) at (2.45,0);

  \fill (B1) circle (1.2pt); \node[left=3pt of B1] {1};
  \fill (B2) circle (1.2pt); \node at ($(B2)+(0,0.24)$) {2};
  \fill (B3) circle (1.2pt); \node at ($(B3)+(0,-0.24)$) {3};
  \fill (B4) circle (1.2pt); \node at ($(B4)+(0.26,0)$) {4};
  \fill (BR) circle (1.2pt);

  \draw (B1) -- node[edgelabel,above,pos=.55,xshift=-0.25cm]{A} (B2);
  \draw (B1) -- node[edgelabel,below,pos=.52,xshift=-7pt]{C} (B3);
  \draw (B1) -- node[edgelabel,above,pos=.55,yshift=1.5pt]{B} (B4);

  \draw (B4) to[bend left=55]  node[edgelabel,above,pos=.45,yshift=2pt]{C} (BR);
  \draw (B4) to[bend right=55] node[edgelabel,below,pos=.45,yshift=-2pt]{A} (BR);

  \draw (B2) to[out=0,in=120]
        node[edgelabel,above,pos=.50,yshift=4pt]{\small cb} (BR);
  \draw (B3) to[out=0,in=-120]
        node[edgelabel,below,pos=.50,yshift=-4pt]{\small ba} (BR);
\end{scope}

\begin{scope}[shift={(2*\DX,0)},local bounding box=G3]
  \coordinate (C1) at (0,0);
  \coordinate (C2) at (1, \H/2);
  \coordinate (C3) at (1,-\H/2);
  \coordinate (C4) at (2.0,0);

  \fill (C1) circle (1.2pt); \node[left=4pt of C1] {1};
  \fill (C2) circle (1.2pt); \node at ($(C2)+(0,0.24)$) {2};
  \fill (C3) circle (1.2pt); \node at ($(C3)+(0,-0.24)$) {3};
  \fill (C4) circle (1.2pt); \node at ($(C4)+(0.22,0)$) {4};

  \draw (C1) -- node[edgelabel,above,pos=.45,xshift=-6pt]{A} (C2);
  \draw (C3) -- node[edgelabel,below,pos=.65,yshift=-6pt]{\small ba} (C4);

  \draw (C1) to[bend left=52]
        node[edgelabel,left,pos=.38,xshift=17pt]{B} (C3);
  \draw (C1) to[bend right=52]
        node[edgelabel,left,pos=.58,xshift=-7pt]{C} (C3);

  \draw (C2) to[bend left=48]
        node[edgelabel,right,pos=.42,xshift=5pt]{C} (C4);
  \draw (C2) to[bend right=32]
        node[edgelabel,right,pos=.60,xshift=-16pt]{A} (C4);
\end{scope}

\path (G1.east) -- (G2.west) node[midway] (mid12) {};
\path (G2.east) -- (G3.west) node[midway] (mid23) {};
\node at (mid12) {$\Theta$};
\node at (mid23) {$\Theta$};

\node at ([yshift=-\LBL]G1.south) {(a)};
\node at ([yshift=-\LBL]G2.south) {(b)};
\node at ([yshift=-\LBL]G3.south) {(c)};

\end{tikzpicture}
\caption{The transformation of K4}
\label{fig:k4}
\end{figure}

\begin{itemize}
  \item On branch \(2\!-\!3\), apply the \(\Delta\!\leftrightarrow\!Y\) transform,
        i.e., insert an auxiliary node and label the new branches \texttt{cb}
        and \texttt{ba}, whose series combination equals \texttt{bb} (that is, \(B\)).
  \item In Fig.~\ref{fig:k4}(b), by the modular distribution law, obtain the
        parallel branches \texttt{C} and \texttt{A} between node \(4\) and the
        newly created unnamed node.
  \item The branch \(\texttt{CA}\) (the parallel combination of
        \texttt{C} and \texttt{A}) can be ``slid'' into node \(2\);
        the artificially created branch \texttt{cb} then disappears
        (it is absorbed into node \(2\)).
  \item By ``sliding'' branch \texttt{B} \((1\!-\!4)\) across branch
        \texttt{ba}\,, we obtain the final graph (Fig.~\ref{fig:k4}(c))
        in which the impedance (jorb) between any pair of nodes can be computed.
\end{itemize}

After the transformation, the equivalent impedances (jorbs) are:
\begin{flalign*}
& \vartheta(1,4) = (A \eqcirc CA) \pl (ba \eqcirc CB) = ACACB, &&\\
& \vartheta(1,2) = BCACA, \qquad
  \vartheta(1,3) = CBACA, &&\\
& \vartheta(2,3) = BCACA, \qquad
  \vartheta(2,4) = CBACA, \qquad
  \vartheta(3,4) = BCACA, &&\\
& \vartheta(1,1) = ca \pl (A \eqcirc CA \eqcirc ba \eqcirc CB)
                 = CACA = J_s(x), \qquad \lambda_s(x) = 4. &&
\end{flalign*}
\end{example}

\newcommand{\oxB}{%
  \begin{circuitikz}[american, scale=0.75]
    \ctikzset{bipoles/length=12mm}
    \draw[draw=none] (-0.3,0.60) rectangle (2.9,-0.60);
    \draw (0,0) to[short,o-] ++(3mm,0)
          to[R] ++(12mm,0)
          to[short,-o] ++(3mm,0);
  \end{circuitikz}%
}

\newcommand{\oxBC}{
  \begin{circuitikz}[american, scale=0.75]
    \ctikzset{bipoles/length=12mm}
    \draw[draw=none] (-0.3,0.60) rectangle (2.9,-0.60);
    \draw (0,0) to[short,o-] ++(3mm,0)
          to[R]            ++(12mm,0)
          to[short]        ++(2mm,0) coordinate (J) 
          node[ocirc]{}                         
          to[short]        ++(2mm,0)            
          to[L]            ++(12mm,0)
          to[short,-o]     ++(3mm,0);
  \end{circuitikz}%
}

\newcommand{\oxCA}{
  \begin{circuitikz}[american, scale=0.75]
    \ctikzset{bipoles/length=10mm}
    \draw[draw=none] (-0.3,0.80) rectangle (2.9,-0.60);
    \draw (0,0) to[short,o-] (0.2,0) coordinate (L);
    \coordinate (R) at (2.4,0);
    \draw (R) to[short,-o] (2.6,0);
    \draw (L) -- ++(0,0.4) coordinate (Lt)
          (L) -- ++(0,-0.4) coordinate (Lb)
          (R) -- ++(0,0.4) coordinate (Rt)
          (R) -- ++(0,-0.4) coordinate (Rb);
    \draw (Lt) to[L] (Rt);
    \draw (Lb) to[C] (Rb);
  \end{circuitikz}%
}

\newcommand{\oxCB}{
  \begin{circuitikz}[american, scale=0.75]
    \ctikzset{bipoles/length=10mm}
    \draw[draw=none] (-0.3,0.85) rectangle (2.9,-0.60);
    \draw (0,0) to[short,o-] (0.2,0) coordinate (L);
    \coordinate (R) at (2.4,0);
    \draw (R) to[short,-o] (2.6,0);
    \draw (L) -- ++(0,0.5) coordinate (Lt)
          (L) -- ++(0,-0.5) coordinate (Lb)
          (R) -- ++(0,0.5) coordinate (Rt)
          (R) -- ++(0,-0.5) coordinate (Rb);
    \draw (Lt) to[R] (Rt); 
    \draw (Lb) to[L] (Rb); 
  \end{circuitikz}%
}

\newcommand{\oxC}{%
  \begin{circuitikz}[american, scale=0.75]
    \ctikzset{bipoles/length=12mm}
    \draw[draw=none] (-0.3,0.60) rectangle (2.9,-0.60);
    \draw (0,0) to[short,o-] ++(3mm,0)
          to[L] ++(12mm,0)
          to[short,-o] ++(3mm,0);
  \end{circuitikz}%
}




\section{Fundamentals of M‑Logic}

M‑space admits a family of higher‑order operators, most of which can be viewed as \emph{permutations}. Without delving into the full algebraic machinery—such as the $G$‑ and $H$‑operators or the eight‑element permutation groups—one permutations merit particular attention: the \emph{complement} permutation
  \[
     K(x)\;=\;\ell'_x\,x\,r'_x.
  \]
The permutation \(K(\,)\) acts as the \emph{negation operator} in the
M‑number formulation of formal logic.  
Recall that the shell \(q(\,)\) plays the rôle of a logical variable in
M‑logic; hence the entire many‑valued calculus can be developed by
re‑using the algebraic identities already obtained for electrical‑network
jorbs.

Under this interpretation

\begin{itemize}
  \item a \emph{series} composition of jorbs corresponds to logical
        \textsf{AND};
  \item a \emph{parallel} composition corresponds to logical
        \textsf{OR};
\end{itemize}

only, instead of operating on the complete jorb, one works with its
boundary components \(l_x\) and \(r_x\) (and, when convenient, with
their duals).

\medskip
The core logical operators expressed in terms of
M‑shells are summarised in Table~\ref{tab:Mlogic}.

\begin{table}[h]
  \centering
  \caption{Basic logical operators in M‑shell notation}
  \label{tab:Mlogic}
  \begin{tabular}{@{}lll@{}}
    \toprule
    Operator & Symbolic form & M‑shell realisation \\ \midrule
    Conjunction (\textsc{and})      & \(x \land y\) &
      \(x \eqcirc y\) (series) \\
    Disjunction (\textsc{or})       & \(x \lor y\)  &
      \(x \pl y\) (parallel) \\
    Negation (\textsc{not})         & \(\lnot x\)   &
      \(K(x)=\ell'_x \,x\, r'_x\) \\
    Identity (logical variable)     & \(x\)         &
      \(q(x)\) (shell) \\
    Tautology / Contradiction       & \(1,\;0\)     &
      \( s‑zero, p‑zero \)\quad over $\Gamma$ \\ \bottomrule
  \end{tabular}
\end{table}
Thus we obtain concise, formal expressions (Table ~\ref{tab:logic}) that are amenable to machine processing of formalised logic:

\begin{table}[ht]
\centering
\captionsetup{skip=7pt}
\caption{$m$-logical expression}
\label{tab:logic}
\begin{tabular}{l|c|l}
\hline
$operator$        & symbol      & expression for logical operations           \\
\hline
Nnegation      & $\neg$      & $l_x' \cdot r_x'$   \\
Conjunction      & $\wedge$      & ($l_x \downarrow l_y) \cdot (r_x \uparrow r_y $)   \\
Disjunction      & $\vee$      & ($l_x \uparrow l_y) \cdot (r_x \downarrow r_y $)   \\
p-implication       & $\overset{p}{\longrightarrow}$      & ($l_x' \uparrow l_y) \cdot (r_x' \downarrow r_y $)   \\
s-implication      & $\overset{s}{\longrightarrow}$      & ($l_x' \downarrow l_y) \cdot (r_x' \uparrow r_y $)   \\
p-equivalence      & $\overset{p}{\Longleftrightarrow}$ & $((l_x' \uparrow l_y) \downarrow (l_x \uparrow l_y' )) \cdot ((r_x' \downarrow r_y) \uparrow (r_x \downarrow r_y' ))$  
\\
s-equivalence       & $\overset{s}{\Longleftrightarrow}$ & ($(l_x' \downarrow l_y) \uparrow (l_x \downarrow l_y')) \cdot ((r_x' \uparrow r_y) \downarrow (r_x \uparrow r_y' )$) \\
\hline
\end{tabular}
\end{table}

Instead of computing the already mentioned equation (eq. \ref{eq:ser}):
\[x \eqcirc y := (l_x \downarrow l_y)\cdot (l_x \downarrow l_y)\cdot x \cdot (r_x \uparrow l_y)\cdot (r_x \downarrow l_y)\cdot y \cdot (r_x \uparrow r_y)\cdot (r_x \uparrow r_y),
\]

in $m$-logic it is sufficient to calculate
\begin{equation}
x \wedge y := (l_x \downarrow l_y)\cdot (r_x \uparrow r_y),
\end{equation}
because the shell of the expression \(e = x \wedge y\) is \(q(e)=l_e \cdot r_e\).\\
As previously shown, \((aa)\eqcirc (bb) = aaaababbbb\), which represents the variable \(ab\) in $m$-logic, since \(q(aaaababbbb) = ab\).
If one is working in two-valued (Boolean) logic (the \(\Gamma_2\) alphabet), then ‘ba’ represents truth (\(\top\)) and ‘ab’ falsehood (\(\bot\)).
Clearly, it is easier to use the Table ~\ref{tab:logic}:
\[
(aa)\wedge(bb) = (a \downarrow b)\cdot (a \uparrow b) = ab.
\]

Analogously, for the logical OR function (i.e. parallel connection, \(\pl\)):
\[
(aa)\vee(bb) = (a \uparrow b)\cdot (a \downarrow b) = ba.
\]
Note that with the two-letter alphabet $\Gamma_2$ there are only four possible values of $q()$. (In Dunn–Belnap’s four-valued logic $B4$, the pair 'aa' corresponds to $B$ and 'bb' to $N$.) With the three‑letter alphabet $\Gamma_3$ there are nine possible values of $q$, precisely the arity found in the early twentieth-century three-valued
 calculi.\footnote{For the three-valued case, Šare's M‑logic coincides with Kleene's $K_1$; nonetheless, its generalisation reaches well beyond the traditional two- and three-valued frameworks of classical logic.}
 
Just as conjunction and disjunction are dual, there also exist p- and s-implications and equivalences.

The traditional implication (here, p-implication) \(P \rightarrow Q = \neg P \wedge Q\) becomes, in $m$-logic,
\begin{equation}
K(P)\wedge Q = (l_p' \uparrow l_q)\cdot (r_p' \downarrow r_q).
\end{equation}

Figure \ref{hasseab} presents the Hasse diagram of the partial order
$(\le_q)$ on the set of equivalence classes of CRL two‑terminals.
The algebra of residual classes modulo unity is obtained by introducing
two binary operations—series and parallel connection—on the classes.
This algebra forms a \emph{lattice}: the series connection corresponds to
the \emph{supremum}, and the parallel connection to the \emph{infimum}.
The lattice is isomorphic to the cardinal product of anti‑isomorphically
ordered three‑element chains.

Special attention should be paid to the peculiar \emph{logic} that appears in
classifying CRL two‑terminals.  
If a truth function is considered on the shell of a two‑terminal, the logic
is characterised by nine pairs of truth values.  
Accepting the symmetry of dual physical tests leads to a nine‑valued logic
of electrical two‑terminals:
\begin{itemize}
\item $\top$\;: voltage is zero under constant current,
\item $\bot$\;: current is zero under constant voltage,
\item $\mid$\;: neither condition holds.
\end{itemize}

The symbols $\top$, $\bot$, and $\mid$ are read \emph{true}, \emph{false}, and \emph{neither true nor false}, respectively, replicating the practice
used in logic circuits.


\newcommand{\ValFirst}[1]{%
  \ifcase#1\relax\or $\bot$\or $\mid$\or $\top$\fi}

\newcommand{\ValSecond}[1]{%
  \ifcase#1\relax\or $\top$\or $\mid$\or $\bot$\fi}

\newcommand{\PlaceCoords}{%
  \foreach \i in {1,...,3}{%
    \foreach \j in {1,...,3}{%
      \coordinate (p\i\j) at ({\i+\j},{\j-\i});%
    }%
  }}

\newcommand{\DrawEdges}{%
  \foreach \i in {1,2}{%
    \foreach \j in {1,2,3}{%
      \draw[line width=0.9pt] (p\i\j) -- (p\the\numexpr\i+1\relax\j);%
    }}%
  \foreach \i in {1,2,3}{%
    \foreach \j in {1,2}{%
      \draw[line width=0.9pt] (p\i\j) -- (p\i\the\numexpr\j+1\relax);%
    }}}

\begin{figure}[ht]
\centering
\begin{minipage}[t]{0.47\textwidth}
\centering
\begin{tikzpicture}[scale=1, line cap=round, line join=round]
  \PlaceCoords
  \DrawEdges
  \foreach \i in {1,...,3}{%
    \foreach \j in {1,...,3}{%
      \node[font=\small, fill=white, inner sep=1pt]
        at (p\i\j) {[{\Ltr{\i}}{\Ltr{\j}}]};%
    }}%
\end{tikzpicture}

\vspace{0.5ex}{\small (a)}
\end{minipage}
\hfill
\begin{minipage}[t]{0.47\textwidth}
\centering
\begin{tikzpicture}[scale=1, line cap=round, line join=round]
  \PlaceCoords
  \DrawEdges
  \foreach \i in {1,...,3}{%
    \foreach \j in {1,...,3}{%
      \node[font=\small, fill=white, inner sep=1pt]
        at (p\i\j) {[\,\ValFirst{\i},\,\ValSecond{\j}\,]};%
    }}%
\end{tikzpicture}

\vspace{0.5ex}{\small (b)}
\end{minipage}
\caption{Hasse diagram for the three-generator lattice: a) jorb; b) logic.}
\label{hasseab}
\end{figure}
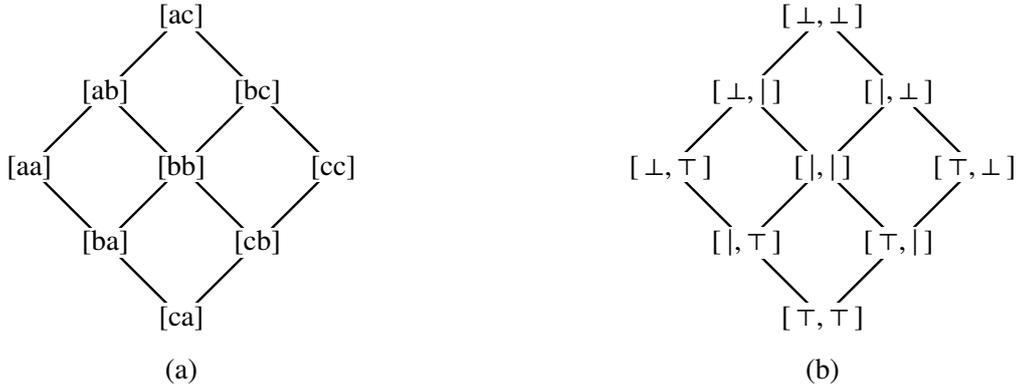

\section{Jorb and the Impedance of Electrical Networks}
The problem of establishing equivalence between jorbs—particularly when calculations are carried out manually—requires the use of the compression or expansion function. For example, while the jorbs CABA, CAbcaa, CAbcbbA, and CBABA are equivalent, this is not trivial to demonstrate. The equivalence problem is addressed by introducing a suitable metric, specified by two parameters: $\lambda_s$ and $\lambda_p$.

These parameters reveal a fundamental relationship between the valuation of a jorb shell - the defect $\Delta(x)= v(r_x) - v(l_x)$ from (eq. \ref{eq:def}), and the jorb’s $\lambda$-length (eq. \ref{lam}), as expressed by the following formulas:
\begin{equation}
\lambda_s(x) = \frac{\lambda(x) - \Delta(x)}{2}, \qquad
\lambda_p(x) = \frac{\lambda(x) + \Delta(x)}{2}
\label{eq: lsp}
\end{equation}
It is easy to see that:
\begin{equation}
\lambda_s(x) - \lambda_p(x) = v(l_x) - v(r_x)
\label{eq: lssp}
\end{equation}

The valuation of a jorb, for any alphabet \(\Gamma\), is a function \(\Phi_\Gamma\)\footnote{The condition is that the valuation of the alphabet is symmetric with respect to the integer 0 (zero).} from the set of jorbs \(\mathbf{M}_\Gamma\) to the m-quadruple, whose components represent the valuation of the start jorb's symbol, s-length, p-length, and valuation of the end jorb's symbol, respectively.
\begin{equation}
\Phi_\Gamma(x) = (v(l_x),\; \lambda_s,\; \lambda_p, \; v(r_x))
\end{equation}

\begin{example}\label{ex:jorb}
\leavevmode
\begin{enumerate}
    \item For \(\Gamma = \{a, b\}, \quad \;\;\; \Phi_\Gamma(BABA) = (1, 4, 2, -1) \)
   \\ For \(\Gamma = \{a, b, c\}, \quad \Phi_\Gamma(BABA)=(0,2,1,-1)
    \)

  \item For \(\Gamma = \{a, b, c\}, \\\Phi_\Gamma(CABA)=(1,3,1,-1)\)\\
  \(\Phi_\Gamma(CAbcaa)=(1,3,1,-1)\)\\
  \(\Phi_\Gamma(CAbcbbaa)=(1,3,1,-1)\).
\end{enumerate}

\end{example}

By knowing any three of the four components in the quadruple—since the fourth can always be determined from the other three—the jorb is uniquely defined by its triplet. However, for better clarity of the jorb characteristics, we will retain the quadruple form. Jorbs are equivalent if they have the same quadruples (or triplets) for the same alphabet.

A particularly important relationship of the m-quadruple is with the impedance Z(s) of the electrical network:

For every \(x \in M_\Gamma\),   
\begin{equation}
(v(l_x), \lambda_s, \lambda_p, v(r_x)) \iff Z(s) = 
s^{v(l_x)}\,
\dfrac{\displaystyle\sum_{i=0}^{\lambda_s} A_i s^{i}}
      {\displaystyle\sum_{j=0}^{\lambda_p} B_j s^{j}}  \label{rik2}
\end{equation}
where $A_i$ and $B_j$ are the coefficients of the polynomials in the numerator and denominator, respectively, of the electrical impedance in the complex s-plane.

Within the rational function representation of electrical impedance, $\lambda_s$ and $\lambda_p$ denote the maximal degrees of the numerator and denominator polynomials, respectively.

This condition constitutes the fundamental requirement of the algorithm for generating jorbs and for constructing their classes with a common q-shell, as well as for the synthesis of electrical networks from jorbs.
\begin{example}\label{ex:triplets}
\leavevmode
The $m$-numbers from Figure \ref{fig:goran}  for $\Gamma=\{a, b, c\}$ and $val = \{-1, 0, 1\}$ have the following triplets: ABAB = (-1, 1, 0), CBCB = (1, 2, 0), BABA = (0, 2, -1), BCBC = (0, 1, 1), BACB = (0, 2, 0), BCAB = (0, 2, 0) .
\end{example}
\section{Classification and Categorization of Electrical Networks}

In the context of electrical network analysis, the terms \textit{classification} and \textit{categorization} denote related yet conceptually distinct processes. \textit{Classification} refers to the systematic organization of networks according to well-defined analytical or topological criteria, typically resulting in a hierarchical structure-for example, according to their topology. 

On the other hand, \textit{categorization} involves grouping networks based on shared or similar properties, often without imposing a strict hierarchical framework. This process emphasizes the recognition of structural or functional patterns—for example, grouping networks that share identical $q$-shells, equivalent impedance forms, or symmetry properties. 

Thus, while classification establishes a formal taxonomy grounded in predefined rules, categorization highlights the relational and often emergent similarities among networks within a broader analytical framework.

\subsection{Jorbs Categorization Based on Landenheim's Catalogue}
The first systematic attempt to study electrical networks by means of exhaustive enumeration was undertaken in E. Ladenheim’s Master’s thesis \cite{ladenheim1948synthesis} in 1948, under the supervision of Ronald Martin Foster. Alongside Wilhelm Adolf E. Cauer, Foster is regarded as a founding figure of electrical network synthesis (1924). 

The \textit{Ladenheim catalogue} comprises the complete set of electrical networks containing at most two energy storage elements (inductors or capacitors) and no more than three resistors.
Although earlier attempts at cataloguing had been made (\cite{Jiang2011}, \cite{JiangSmith2012}), the most comprehensive work in this regard was carried out by A. Morelli in his doctoral dissertation \cite{phdthesis}.

\noindent
\begin{figure}[h]
\begin{minipage}[t]{0.48\textwidth}
\begin{tabular}{|r|l|p{3cm}|}
\hline
\rowcolor{gray!25}\textbf{No.} & \textbf{Jorb} & \textbf{Scheme label (\#)} \\
\hline
1  & A   & 2 \\
2  & B   & 3 \\
3  & C   & 1 \\
\hline
\rowcolor{gray!20} $\Sigma_1$ & 3 & 3 \\
\hline
4  & AB  & 5 \\
5  & AC  & 4 \\
6  & BA  & 9 \\
7  & BC  & 6 \\
8  & CA  & 7 \\
9  & CB  & 8 \\
\hline
\rowcolor{gray!20} $\Sigma_2$ & 6 & 6 \\
\hline
10 & ABA & 13, 14 \\
11 & ABC & 10 \\
12 & ACB & 26 \\
13 & BAB & 17, 18 \\
14 & BAC & 34 \\
15 & BCA & 49 \\
16 & BCB & 15, 16 \\
17 & CAB & 41 \\
18 & CBA & 19 \\
19 & CBC & 11, 12 \\
\hline
\rowcolor{gray!20} $\Sigma_3$ & 10 & 14 \\
\hline
\end{tabular}
\end{minipage}%
\hspace{.01cm}%
\begin{minipage}[t]{0.48\textwidth}
\begin{tabular}{|r|l|p{5cm}|}
\hline
\rowcolor{gray!25}\textbf{No.} & \textbf{Jorb} & \textbf{Scheme label (\#)} \\
\hline
20 & ABAB  & 20, 21, 22, 23 \\
21 & ABCB  & 24, 25 \\
22 & BABA  & 43, 44, 45, 46 \\
23 & BABC  & 32, 33 \\
24 & BACB  & 42, 62, 71, 72, 88, 96, 97\\
25 & BCAB  & 27, 63, 73, 74, 87 \\
26 & BCBA  & 47, 48 \\
27 & BCBC  & 28, 29, 30, 31 \\
28 & CBAB  & 39, 40 \\
29 & CBCB  & 35, 36, 37, 38 \\
\hline
\rowcolor{gray!20} $\Sigma_4$ & 10 & 36 \\
\hline
30 & BABAB & 75, 76, 77, 78, 79, 80, 81, 82, 83, 84, *85, *86 \\
31 & BABCB & 64, 65, 66, 89, 90, 91, 98, 100, 103, 106 \\
32 & BCBAB & 67, 68, 69, *70, 92, 93, 94, *95, 99, 101, 102, 104, *105, *107, *108 \\
33 & BCBCB & 50, 51, 52, 53, 54, 55, 56, 57, 58, 59, *60, *61 \\
\hline
\rowcolor{gray!20} $\Sigma_5$ & 4 & 49 \\
\hline
\rowcolor{gray!35}{\Large $\Sigma$} & \textbf{33} & \textbf{108} \\
\hline
\end{tabular}
\end{minipage}
\caption{Jorb's classes of Ladenheim's networks}
\label{fig:JorbClass}
\end{figure}
Morelli's novelty was that he cleaned and corrected Ladenheim's raw list, introduced group-theoretic classification ($s$: frequency inversion, $d$:impedance inversion, $p=s \cdot d$ composition), and then went further by grouping topologies into equivalence classes of network functions — ultimately showing that only 35 essentially distinct types remain out of the original 148, i.e. $148 \rightarrow 108 \rightarrow 62 \rightarrow 35$.
 
The table shown in Figure \ref{fig:JorbClass} was obtained by applying the $m$-theory to Morelli’s list of electrical networks provided in Appendices A and B \cite[see p.~175]{phdthesis}. For each of the 108  schematics, the corresponding jorb was calculated and then categorized according to the number of elements\footnote{Although $m$-theory is defined for alphabets with an arbitrary number of symbols, this work will focus exclusively on the alphabet $\Gamma_3$, in order to preserve a clear physical interpretation. This alphabet corresponds to the three electrical elements C, R, and L, which are associated with the letters A, B, and C from the abstract alphabet $\{a, b, c\}$, with valuations $\{-1, 0, 1\}$.}. For each such class, the corresponding table lists both the total number of jorbs belonging to the class and the number of original schemes associated with it. The third column lists the networks originally labeled '$\#n$' and indicates which Ladenheim/Morelli networks share the same jorb.

The scheme numbers in Ladenheim's catalog that represent bridge connections in the table (fig.  \ref{fig:JorbClass}) are marked with an asterisk (`*`) in front of the number. For example, the jorb $BABAB$ has SP schemes \#75, \#76, \ldots, \#84  and the bridge schemes \#85 (*85) and \#86 (*86)\footnote{It can be shown, by means of the $\Delta$–$Y$ topological transformation (fig. \ref{fig:delta}), that the bridge network, e.g., $*60$ becomes the SP network $\#59$, and the network $*85$ becomes $\#83$. With similar interventions, Morelli's classification can be partially modified.
}.  

Finally, the table reports the total number of distinct jorbs obtained, together with the number of original schemes from the Morelli catalog.

It should be noted that the table illustrates only the correspondence between the Ladenheim/Morelli schematics and the jorb-based representation, without asserting its completeness. For instance, the jorbs \texttt{BACB} and \texttt{BCAB} describe equivalent jorbs; consequently, in the final classification the corresponding schemes are consolidated under a single representative quadruple, thereby reducing the total number of distinct classes. Moreover, within this revised framework, it becomes unnecessary to impose explicit constraints on the number of reactive components or on the total element count.

\subsection{Computational Generation of Jorbs}
Jorbs represent \emph{non-repetitive strings} (two identical uppercase letters cannot appear consecutively, since the \texttt{zip} function reduces them to a single one). Their computational generation involves the initial and final letter, as well as the total length — i.e., the number of uppercase letters. For an alphabet of size $k$ (e.g., for $\Gamma_3$ we have $k = 3$) and a string of length $n$, the number of possible strings without consecutive identical uppercase letters is given by:
\begin{equation}
N(n) = k \cdot (k - 1)^{n - 1}
\end{equation}

For $k = 3$, this means that if $n = 3$ we have $12$ generated jorbs, and for $n = 4$ we have $24$ generated jorbs. In Table \ref{tab:gliste-koliko}, the obtained jorbs are shown together with their quadruples.

\begin{longtable}{|c|c|c|c|}
\caption{Jorbs generated for $k = 3$, $n = 3$ and $n = 4$} \label{tab:gliste-koliko} \\
\hline
  & A & B & C \\
\hline
\endfirsthead
\hline
  & A & B & C \\
\hline
\endhead
A & \begin{tabular}[t]{@{}l@{}}ABA, (-1, 1, 1, -1) \\ ACA, (-1, 2, 2, -1)\end{tabular} & \begin{tabular}[t]{@{}l@{}}ACB, (-1, 1, 2, 0)\end{tabular} & \begin{tabular}[t]{@{}l@{}}ABC, (-1, 0, 2, 1)\end{tabular} \\
\hline
B & \begin{tabular}[t]{@{}l@{}}BCA, (0, 2, 1, -1)\end{tabular} & \begin{tabular}[t]{@{}l@{}}BAB, (0, 1, 1, 0) \\ BCB, (0, 1, 1, 0)\end{tabular} & \begin{tabular}[t]{@{}l@{}}BAC, (0, 1, 2, 1)\end{tabular} \\
\hline
C & \begin{tabular}[t]{@{}l@{}}CBA, (1, 2, 0, -1)\end{tabular} & \begin{tabular}[t]{@{}l@{}}CAB, (1, 2, 1, 0)\end{tabular} & \begin{tabular}[t]{@{}l@{}}CAC, (1, 2, 2, 1) \\ CBC, (1, 1, 1, 1)\end{tabular} \\
\hline
\end{longtable}

\begin{longtable}{|c|c|c|c|}
\hline
  & A & B & C \\
\hline
\endfirsthead

\hline
  & A & B & C \\
\hline
\endhead

\hline
\endfoot

\hline
\endlastfoot
A & \begin{tabular}[t]{@{}l@{}}ABCA, (-1, 2, 2, -1) \\ ACBA, (-1, 2, 2, -1)\end{tabular} & \begin{tabular}[t]{@{}l@{}}ABAB, (-1, 1, 2, 0) \\ ABCB, (-1, 1, 2, 0) \\ ACAB, (-1, 2, 3, 0)\end{tabular} & \begin{tabular}[t]{@{}l@{}}ABAC, (-1, 1, 3, 1) \\ ACAC, (-1, 2, 4, 1) \\ ACBC, (-1, 1, 3, 1)\end{tabular} \\
\hline
B & \begin{tabular}[t]{@{}l@{}}BABA, (0, 2, 1, -1) \\ BACA, (0, 3, 2, -1) \\ BCBA, (0, 2, 1, -1)\end{tabular} & \begin{tabular}[t]{@{}l@{}}BACB, (0, 2, 2, 0) \\ BCAB, (0, 2, 2, 0)\end{tabular} & \begin{tabular}[t]{@{}l@{}}BABC, (0, 1, 2, 1) \\ BCAC, (0, 2, 3, 1) \\ BCBC, (0, 1, 2, 1)\end{tabular} \\
\hline
C & \begin{tabular}[t]{@{}l@{}}CABA, (1, 3, 1, -1) \\ CACA, (1, 4, 2, -1) \\ CBCA, (1, 3, 1, -1)\end{tabular} & \begin{tabular}[t]{@{}l@{}}CACB, (1, 3, 2, 0) \\ CBAB, (1, 2, 1, 0) \\ CBCB, (1, 2, 1, 0)\end{tabular} & \begin{tabular}[t]{@{}l@{}}CABC, (1, 2, 2, 1) \\ CBAC, (1, 2, 2, 1)\end{tabular} \\
\hline
\end{longtable}
\addtocounter{table}{-1}

Note that the jorbs \texttt{ACA} and \texttt{CAC} for $n = 3$, as well as thirteen jorbs from the table for $n = 4$ (\texttt{ABCA}, \texttt{ACBA}, \texttt{ACAB}, \ldots, \texttt{CBAC}), do not appear in Fig.~\ref{fig:JorbClass}. The reason is that each of them contains more than two reactive elements (capacitor~A and/or inductor~C), which was one of the constraints of Ladenheim’s networks.

The limitation of the classification for $n = 5$ and higher is evident from only four (out of six) jorbs found in Fig.~\ref{fig:JorbClass}, specifically in the B--B cell of the overall matrix containing 48 elements. Table ~\ref{tab:cell-BB} shows this cell, for which, due to a known reason, the networks with the jorbs \texttt{BACAB} and \texttt{BCACB} were not processed.

\begin{table}[h!]
\centering
\caption{Elements from cell B–-B (for $n = 5$)}
\label{tab:cell-BB}
\begin{tabular}{|l|c|}
\hline
\textbf{Jorb} & \textbf{Quadruple} \\
\hline
BABAB & (0, 2, 2, 0) \\
BABCB & (0, 2, 2, 0) \\
BACAB & (0, 3, 3, 0) \\
BCACB & (0, 3, 3, 0) \\
BCBAB & (0, 2, 2, 0) \\
BCBCB & (0, 2, 2, 0) \\
\hline
\end{tabular}
\end{table}

Much has been done recently ( \cite{morelli2019passive}, \cite{9027922}, \cite{8967022}, \cite{JiangSmith2012}). However, there is no generalized theory that covers the treated cases. Therefore, the perspective with jorbs gives new hope.

\subsection{Equivalence of Impedances}

The possible impedance equivalence under the $m$-theory formalism is particularly transparent: two networks can be impedance-equivalent if, and only if, their quadruples coincide.
 This equivalence can be exemplified by networks $\#47$ and $\#48$, as shown in equations~(\ref{mat:47-48}). In this case, the \texttt{BCBA} jorb, i.e. (0, 2, 1, -1) quadruple, establishes the possible equivalence between the networks. The equivalence arises from distinct topologies producing identical $Z(s)$, provided that the RLC component values satisfy the substitution relations given in equation~(\ref{eq:eqparams-47-48}), transforming $Z_{47}(s)$ into $Z_{48}(s)$.

\begin{equation}
\begin{tabular}{|>{\centering\arraybackslash}m{0.33\linewidth}|>{\centering\arraybackslash}m{0.48\linewidth}|}
\hline
\begin{minipage}[t]{\linewidth}
  \centering
  \includegraphics[width=0.6\linewidth]{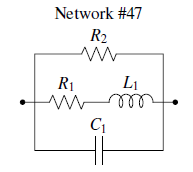}\\[4pt]
  $Z_{47} = (R_2 \; p \; (R_1 \; s \; L_1) \; p \; C_1)$
\end{minipage}
&
$\displaystyle
Z_{47}(s)= \frac{R_1 R_2 + s\,R_1 L_1}
{(R_1+R_2)\;+\;s\,(L_1 + C_1 R_1 R_2)\;+\;s^{2}\,C_1 R_1 L_1}
$
\\ \hline
\begin{minipage}[t]{\linewidth}
  \centering
  \includegraphics[width=0.6\linewidth]{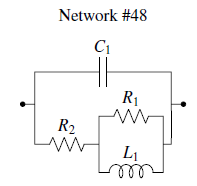}\\[4pt]
  $Z_{48} = (C_1 \; p \; (R_2 \; s \; (R_1 \; p \; L_1)))$
\end{minipage}
&
$\displaystyle
Z_{48}(s)= \frac{R_1 R_2 + s\,L_1\,(R_1+R_2)}
{R_1 \;+\; s\,(L_1 + C_1 R_1 R_2)\;+\; s^{2}\,C_1 L_1 (R_1+R_2)}
$
\\ \hline
\end{tabular}
\label{mat:47-48}
\end{equation}
\noindent
Equivalent parameters ensuring $Z_{48}'(s) \equiv Z_{47}(s)$ as \textbf{BCBA} jorb are:
\begin{equation}
\boxed{
R_1'=\frac{R_1^{2}}{R_1+R_2},\quad
R_2'=\frac{R_1R_2}{R_1+R_2},\quad
L_1'=\frac{L_1\,R_1^{2}}{(R_1+R_2)^{2}},\quad
C_1'=C_1.}
\label{eq:eqparams-47-48}
\end{equation}

An analogous equivalence holds for schemes $\#39$ and $\#40$ (see No. 28 in Fig. \ref{fig:JorbClass}), which share the same jorb \texttt{CBAB}. Conversely, any attempt to equate impedances corresponding to different jorbs—for example, networks $\#47$ and $\#39$—is inherently inconsistent, since the jorb encodes the underlying structural composition of the network. This distinction is, for instance, reflected in the degree of the numerator polynomial of $Z(s)$\footnote{It can be shown that schemes $\#39$ and $\#40$ possess numerators of $Z(s)$ of degree 2, whereas schemes $\#47$ and $\#48$ exhibit numerators of degree 1, precluding their impedance equivalence.}.

In determining network equivalence through the \( m \)-theory, one can discern the behavioral theory of Jan Willems \cite{Willems1991Behavioral}.  
As an example, consider two dual JORBs, \( BAB \) and \( BCB \), each having two topological realizations (TR):

\[
\begin{aligned}
BAB &: 
\begin{cases}
\mathrm{TR}_{11} = \big((C_1 \, s \, R_1) \, p \, R_2\big), \\
\mathrm{TR}_{12} = \big((C_1 \, p \, R_1) \, s \, R_2\big),
\end{cases} \\[6pt]
BCB &: 
\begin{cases}
\mathrm{TR}_{21} = \big((L_1 \, s \, R_1) \, p \, R_2\big), \\
\mathrm{TR}_{22} = \big((L_1 \, p \, R_1) \, s \, R_2\big).
\end{cases}
\end{aligned}
\]

Although identical in the structure of \( Z(s) \), attempts to obtain coefficient values that would demonstrate impedance equivalence—for example, between the first realizations 
\( \mathrm{TR}_{11} \) and \( \mathrm{TR}_{21} \)—prove unsuccessful.  
However, if one of them is taken in the reciprocal form \( 1 / Z(s) \), equivalence is achieved.  
This implies that the impedance \( Z(s) \) and the admittance \( 1 / Z(s) \) represent the same behavior, that is, the same - \emph{imitance}.

We therefore regard the jorb-based representation of impedances as a novel analytical perspective in network analysis and synthesis. Given the structural analogy among various physical domains—electrical, mechanical, fluidic, and thermodynamic—we anticipate that this representation will also find applications beyond electrical network theory.\footnote{In particular, the emergence of the inerter (\cite{annurev}, \cite{MA2021112655}) has reinforced the cross-disciplinary connections of one-port network theory with other branches of engineering.}

\section{Synthesis of SP Electrical Networks}
The derivation of electrical circuits of different topologies from a transfer function or impedance \(Z(s)\) has been a long-standing problem for almost a century. So far, no algorithmic approach exists that provides a complete classification of all topological equivalents, and the question of an optimal algorithm for minimal topological realizations (with the smallest number of elements) still remains open (see \cite{Youla2016} for related discussions).

\subsection{Network Synthesis in Cauer–Foster Forms}

The synthesis of electrical networks is traditionally carried out in the
four canonical Cauer/Foster forms (see \cite{ChenSmith2009} and \cite{JiangSmith2012}).  Each form has its dual:
\emph{Cauer I} versus \emph{Cauer II}, and \emph{Foster I} versus
\emph{Foster II}.  When the ladder is written in jorb notation,
the synthesis becomes straightforward to implement, as illustrated below
for Cauer I and Cauer II.  Note, however, that this procedure produces a
\emph{single}—rather than an exhaustive—realisation with respect to the
network topology.

\begin{table}[h]
  \centering
  \caption{Canonical forms and their ladder notation}
  \label{tab:canonical-ladders}
  \begin{tabular}{@{\!\!}ll@{}}
    \toprule
    \textbf{Form} & \textbf{Ladder notation} \\ \midrule
    Cauer I   & $Z_{1}s\bigl(Y_{2}p\bigl(Z_{3}s\bigl(Y_{4}p(\,\dots)\bigr)\bigr)\bigr)$ \\
    Cauer II  & $Y_{1}p\bigl(Z_{2}s\bigl(Y_{3}p\bigl(Z_{4}s(\,\dots)\bigr)\bigr)\bigr)$ \\
    Foster I  & $p\bigl(R_{0},\,RL_{1},\,RC_{2},\,RLC_{3},\,\dots\bigr)$ \\
    Foster II & $R_{\infty}s\bigl(CR_{1}s\bigl(LR_{2}s\bigl(LCR_{3}s(\,\dots)\bigr)\bigr)\bigr)$ \\
    \bottomrule
  \end{tabular}
\end{table}
In Table \ref{tab:canonical-ladders} \(Z_x\) denotes an impedance, \(Y_x\) an admittance, the symbols \(s\) and \(p\) indicate series and parallel connections, respectively, and the circuit elements \(R\), \(L\), and \(C\) appear with their corresponding indices.

For the jorb $BABCB$, the \textit{Cauer I} ladder can be read off directly from Table~\ref{tab:canonical-ladders}\footnote{The asterisk ‘\(*\)’ is used to denote an admittance, since the superscript \(()^{-1}\) would impair readability.
}.
 
\[ \boxed{Z_{\text{BABCB}}^{(\mathrm{C\,I})}(s)=\ \text{\texttt{R1 s(C1* p(R2 s(L1* p(R3))))}}} \]
where:
\[
\begin{array}{c|c|c|c}
\text{Branch}&\text{Type}&\text{Label}&\text{Math.\ form}\\\hline
1&\text{series}&Z_1&R_1\\
2&\text{paralell}&Y_2&C_1s\\
3&\text{series}&Z_3&R_2\\
4&\text{paralell}&Y_4&\dfrac1{L_1s}\\
5&\text{series}&Z_5&R_3\\
\end{array}
\]
which gives:
\[ Z(s) = Z_1 + \dfrac{1}{Y_2 + \dfrac{1}{Z_3 + \dfrac{1}{Y_4 + \dfrac{1}{Z_5}}}} \]

Similarly, for the jorb $BABCB$ the \textit{Cauer II} ladder can be read off directly from Table~\ref{tab:canonical-ladders}.

\[ \boxed{Z_{\text{BABCB}}^{(\mathrm{C\,II})}(s)=\ \text{\texttt{R1* s(C1 p(R2* s(L1 p(R3*))))}}} \]

\begin{example}
In this illustrative example, adapted from \cite{Sare2000} (p.~148), 
we demonstrate, by means of Table~\ref{tab:canonical-ladders}, 
the derivation of all four Cauer/Foster canonical forms for the network 
represented by the jorb~\texttt{ABABAB}\footnote{Notice that $AB$ represents the series connection of a capacitor and a resistor, while $BA$ represents their parallel connection.}.

\begin{center}
\renewcommand{\arraystretch}{1.2} 
\begin{tabular}{@{}ll@{}}
\toprule
\textbf{Form}      & \textbf{Expression} \\ \midrule
Cauer I            & $ABABAB = B\,s\,(A\,p\,(B\,s\,(A\,p\,(B\,s\,A))))$ \\
Cauer II           & $ABABAB = A\,s\,(B\,p\,(A\,s\,(B\,p\,(A\,s\,B))))$ \\
Foster I           & $ABABAB = A\,s\,BA\,s\,BA\,s\,B$ \\
Foster II          & $ABABAB = AB\,p\,AB\,p\,AB$ \\ \bottomrule
\end{tabular}
\end{center}
\end{example}

\subsection{Synthesising SP Topological Configurations}

The other approach is to write a program that generates all topological trees of the specified elements, uses the permutations with ‘s’ (for series) and ‘p’(for parallel) connections as operators, and—after jorb reduction—checks whether the required canonical jorb form with the same quadruple (triplet) is obtained.

For an arbitrary number of components $i,j,k$ with total $ N=i+j+k$, total number of series--parallel networks is:
\[
T_{SP} = \frac{N!}{i!\,j!\,k!}\, C_{N-1}\, 2^{N-1}
\]
where 
$C_{N-1} = \dfrac{1}{N}\binom{2(N-1)}{\,N-1\,}$ is the \textit{Catalan number}. 

Although that number $T_{SP}$ — even for smaller configurations, e.g., two capacitors, three resistors, and one inductor — it becomes 80640, the computer program employs a heuristic algorithm that speeds up the process.

The verification of this program was carried out using a manually obtained solution from \cite{Sare2000}, p.~166, for ABABA jorb, where each circuit was drawn explicitly instead of being represented in the jorb notation.

\begin{multicols}{2}
\begin{enumerate}
\item(((A s B) p A p B) s A)
\item(((A p B) s A s B) p A)
\item((A s B) p (A s B) p A)
\item((A p B) s (A p B) s A)
\item((((A s B) p A) s B) p A)
\item((((A p B) s A) p B) s A)
\item((((A s B) p B) s A) p A)
\item((((A p B) s B) p A) s A)
\item(((A s B) p A) s (A p B))
\item(((A p B) s A) p (A s B))
\end{enumerate}
\end{multicols}

Table \ref{tbl:jorbs} shows the synthesis of the topological realizations for the two jorbs, BABAB and BCBCB, from Table \ref{tab:cell-BB}, corresponding to the Ladenheim networks in Figure \ref{fig:JorbClass}, rows~30 and~33.

\begin{table}[h!]
\centering
\caption{Example of topological realizations.}
\label{tbl:jorbs}
\begin{tabular}{c|l|l}
\hline
\textbf{No.} & \textbf{BABAB} & \textbf{BCBCB} \\
\hline
1 & ((((C1 s R1) p R2) s C2) p R3) & ((((R1 s L1) p R2) s L2) p R3) \\
2 & (((C1 s R1) p C2 p R2) s R3) & (((R1 s L1) p R2 p L2) s R3) \\
3 & (((C1 p R1) s C2 s R2) p R3) & (((R1 p L1) s R2 s L2) p R3) \\
4 & ((((C1 p R1) s R2) p C2) s R3) & ((((R1 p L1) s R2) p L2) s R3) \\
5 & ((C1 s R1) p (C2 s R2) p R3) & ((R1 s L1) p (R2 s L2) p R3) \\
6 & ((C1 p R1) s (C2 p R2) s R3) & ((R1 p L1) s (R2 p L2) s R3) \\
7 & ((((C1 s R1) p C2) s R2) p R3) & ((((R1 s L1) p L2) s R2) p R3) \\
8 & ((((C1 p R1) s C2) p R2) s R3) & ((((R1 p L1) s L2) p R2) s R3) \\
9 & (((C1 p R1) s R2) p (C2 s R3)) & (((R1 p L1) s R2) p (R3 s L2)) \\
10 & (((C1 s R1) p R2) s (C2 p R3)) & (((R1 s L1) p R2) s (R3 p L2)) \\
\hline
\end{tabular}
\end{table}

The program demonstrates the near-complete capability for synthesizing SP topologies of electrical one-port networks, thereby revealing the operational and reactive dualities as well as the mirror symmetries embedded in the jorb representation, and thus providing a foundation for further generalization and refinement of the synthesis approach.


\section{Conclusion}

In this work, we presented a unified \(m\)-theoretic framework for CRL one-port networks that integrates algebra (jorbs, shells/cores, quasi-orders), topology (\(m\)-graphs with the theta map \(\vartheta\), modular distributivity, and the \(\Delta\text{--}Y\) transformation), and logic (series–parallel induced operations). The \(\lambda\text{--}\Delta\) metric defines \(\lambda_s\) and \(\lambda_p\), and, through the valuation morphism \(\Phi\), provides a compact descriptor of the impedance degree structure. The theoretical formulation is supported by constructive methods, including computational generation of jorbs and synthesis workflows for Cauer and Foster ladder forms, as well as exhaustive and heuristic generation of series–parallel topologies. The results demonstrate practical potential for CAD tools, automated model reduction, and formal verification of network equivalence.
\footnote{AI assistants—ChatGPT \cite{openai_chatgpt} and Grok \cite{xai_grok}—were used mainly for programming support and translations. All outputs were reviewed, adapted, and verified by the authors; all conceptual work, derivations, and final interpretations are entirely the authors’ own.}

\section{Future work} 
Among the first tasks currently being carried out are the catalogs of electrical one-port networks, which will be accessible and generated online thanks to the programs developed for this article.

It is evident, even from our tables (fig. \ref{fig:JorbClass}), that the catalog is not complete even under the already imposed constraints (e.g., for jorb classes with 5 elements), and the mixing of SP and bridge networks is also questionable. An attempt will also be made to categorize the networks under improved conditions.
Furthermore, the core of the jorb ($J_s$), which ensures invariance with respect to impedances observed from any two nodes of the network, represents a particular challenge. 

The challenges of applying $m$-numbers in electrical engineering, as well as in related technical fields (mechanics, hydromechanics, fluidics), are of primary importance and are already being pursued. In addition, the extension of the alphabet $\Gamma$ from three to at least three additional symbols is also being investigated, i.e., physical elements for which the $m$-atoms would be: \textbf{D} (dd): \textit{memristor} $M$ with $v = M(q)\, i$ (state-dependent $R$); \textbf{E} (ee): \textit{CPE / Warburg} with $Z = K/(j\omega)^\alpha$ ($0<\alpha<1$); and \textbf{F} (ff): \textit{fractional inductor} with $Z \propto (j\omega)^{+\alpha}$.

\printbibliography[title={References}]
\end{document}